\documentclass[12pt]{amsart}
\usepackage{amsmath, amsthm, latexsym, amssymb, amsfonts, epsf, xcolor, hyperref}
\usepackage{caption,subcaption,enumerate, color}
\usepackage{multirow,multicol}
\usepackage{mathtools}
\usepackage{bm}
\usepackage{blkarray}
\usepackage{booktabs}
\usepackage{framed}
\usepackage{setspace}
\usepackage{tikz}
\usepackage{graphicx}
\usepackage{wrapfig}
\usepackage{float}
\usepackage[dvips]{epsfig}
\usepackage{bbm}
\usepackage[utf8]{inputenc}
\usepackage{url}
\usepackage{doi}
\usepackage{booktabs}
\usepackage{bibentry}
\usepackage[normalem]{ulem}
\usepackage{comment}
\usepackage{lingmacros}
\usepackage{cleveref}
\usepackage{hyperref}

\usepackage{extarrows}
\usepackage{multirow}
\usepackage{mathtools}
\usepackage{enumerate}
\usepackage{array}
\usepackage{bigstrut}
\usepackage{exscale,relsize}
\usepackage{graphicx}

\usepackage[linesnumbered]{algorithm2e}
\usetikzlibrary{shapes.geometric}

\newcommand{\cal}[1]{\mathcal{#1}}
\newcommand{\cA}{\cal A}

\newcommand{\cC}{\cal C}

\newcommand{\cL}{\cal L}

\newcommand{\cS}{\cal S}

\newcommand{\F}{{\mathbb F}}
\newcommand{\N}{{\mathbb N}}
\newcommand{\Z}{{\mathbb Z}}

\def \F{{\mathbb F}}
\def \Z{{\mathbb Z}}

\def\norm#1.#2.{\lVert#1\rVert_{#2}}
 
\newtheorem{theorem}{Theorem} 
\newtheorem{proposition}[theorem]{Proposition} 
\newtheorem{example}[theorem]{Example} 
\newtheorem{lemma}[theorem]{Lemma}

\newtheorem{definition}[theorem]{Definition} 
\newtheorem{remark}[theorem]{Remark}
\newtheorem{corollary}[theorem]{Corollary} 

\theoremstyle{plain}

\theoremstyle{plain}

\newcommand{\acar}{\textrm{ACar}}

\newcommand{\ev}{\operatorname{ev}}
\newcommand{\GRS}{\operatorname{GRS}}
\newcommand{\hull}{\operatorname{Hull}}
\newcommand{\lcm}{\operatorname{lcm}}

\newcommand{\spn}{\operatorname{Span}}
\newcommand{\supp}{\operatorname{supp}}

\newcommand{\TGRSG}{\operatorname{T}}
\newcommand{\tr}{tr}

\newcommand{\hir}[1]{\textcolor{red}{#1}}
\newcommand{\rmv}[1]{}

\begin{document}


\title[Multivariate Goppa codes]{Multivariate Goppa codes}
\author{Hiram H. L\'opez}
\address[Hiram H. L\'opez]{Department of Mathematics\\ Cleveland State University\\ Cleveland, OH USA}
\email{h.lopezvaldez@csuohio.edu}

\author{Gretchen L. Matthews}
\address[Gretchen L. Matthews]{Department of Mathematics\\ Virginia Tech\\ Blacksburg, VA USA}
\email{gmatthews@vt.edu}
\thanks{The first author was partially supported by an AMS--Simons Travel Grant. The second author was supported by NSF DMS-1855136, NSF DMS-2037833, and the Commonwealth Cyber Initiative.}

\begin{abstract}
In this paper, we introduce multivariate Goppa codes, which contain as a special case the well-known, classical Goppa codes. We provide a parity check matrix for a multivariate Goppa code in terms of a tensor product of generalized Reed-Solomon codes. We prove that multivariate Goppa codes are subfield subcodes of augmented Cartesian codes. By showing how this new family of codes relates to tensor products of generalized Reed-Solomon codes and augmented codes, we obtain information about the parameters, subcodes, duals, and hulls of multivariate Goppa codes. We see that in certain cases, the hulls of multivariate Goppa codes (resp., tensor product of generalized Reed-Solomon codes),  are also multivariate Goppa codes (resp. tensor product of generalized Reed-Solomon codes). We utilize the multivariate Goppa codes to obtain entanglement-assisted quantum error-correcting codes and to build families of long LCD, self-dual, or self-orthogonal codes.
\end{abstract}

\maketitle

\section{Introduction}

Goppa codes were introduced in 1971 by V. D. Goppa \cite{Goppa1, Goppa2} using a polynomial $g(x)$, called a generator polynomial, over the finite field $\F_q$ with $q$ elements. Properties of a Goppa code are tied to those of the generator polynomial. For instance, such codes have minimum distance at least $\deg(g) + 1$. Many Goppa codes have parameters exceeding the Gilbert bound. Moreover, Goppa codes have efficient decoding algorithms. The McEliece cryptosystem, of current interest as the basis for one of only remaining candidates in the NIST Post-Quantum Cryptography Standardization \cite{McE_NIST3, McE_NIST1}, employs Goppa codes \cite{McEliece}. Goppa codes can be viewed from several different perspectives, each giving a window into their capabilities. We aim in this work to generalize Goppa codes to a multivariable case.

Let $\mathbb{F}_{q^t}$ be a finite field with  $q^t$ elements. The polynomial ring over $\F_{q^t}$ in $m$ variables is denoted by $\F_{q^t}[x_1,\ldots,x_m]$ or $\F_{q^t}[{\bm x}],$ when there is no ambiguity on the number of variables. A multivariate Goppa code is defined as follows. Fix non-empty subsets $S_1,\ldots,S_m \subseteq \F_{q^t}$ and their {\it Cartesian product} 
\[\cS := S_1\times\cdots\times S_m\subseteq \F_{q^t}^{m}.\]
Enumerate the elements of $\cS=\{\bm{s}_1,\ldots,\bm{s}_n\} \subseteq F_{q^t}^{m}.$ Take $g \in \F_{q^t}[\bm{x}]$ such that $g(\bm{s}_i) \neq 0$ for all $i\in [n].$ In addition, assume that $g$ can be expressed as a product $g=g_1\cdots g_m,$ where $g_i\in \F_{q^t}[x_i].$ The {\it multivariate Goppa code} is denoted and defined by
\[\Gamma(\cal S, g) := \left\{ (c_1,\ldots,c_n) \in \mathbb{F}_q^n : \sum_{i=1}^n \frac{c_i}{\prod_{j=1}^m(x_j-s_{ij})} = 0 \mod g({\bm x}) \right\},\]
where $\bm{s}_i :=(s_{i1},\ldots,s_{im}) \in \cal S$.

Taking $m=1,$ we obtain the Goppa codes as in \cite{Goppa_Berlekamp, Goppa1, Goppa2}. Setting $m=t=1$ gives the codes considered in \cite{GYHZ}. It is worth noting  that $\Gamma(\cal S, g)$ is a code over $\F_q$ of length $n$ given by $\mid \cS \mid$ where $\cS \subseteq \F_{q^t}^{m}$; thus, $n \leq  q^{tm}$. Hence, allowing larger values of $t$ and $m$ provides longer codes over the same field. As we will see in Corollary \ref{mvG_dim} at the end of Section~\ref{mvGoppa_section}, taking larger values of $m$ allows one to obtain codes of the same lengths over the same field but with potentially larger dimensions.

As usual, an $[n,k,d]$ code over $\F_{q^t}$ is a code of length $n$, dimension $k$, and minimum distance $d := \min\{|\supp(c)| : 0\neq c\in C\},$ where $\supp(c)$ denotes the support of $c$, that is, the set of all non-zero entries of $c.$ Given ${v} \in \F^n$, we denote its $i^{th}$ component by $v_i$ where $i \in [n]$. The dual of an $[n,k,d]$ code $C$ is \[C^{\perp}:= \left\{ w \in \F^n: w \cdot c=0 \ \forall c \in C \right\};\] that is, the dual is taken with respect to the Euclidean inner product. The {\it hull} of $C$ is  $\hull(C)=C\cap C^{\perp}.$  The code $C$ is {\it linear complementary dual} ({\it LCD}) \cite{Massey} if $\hull(C)=\{ \bm{0} \}$ and is  {\it self-orthogonal} if $C \subseteq C^{\perp}.$ 

In Section~\ref{mvGoppa_section}, we recall the definition of a generalized Reed-Solomon (GRS) code, which is a well-known code that depends of an integer $k$ and a polynomial $g\in \F_{q^t}[x].$ When $k=\deg(g),$ the GRS code is called a GRS code via a Goppa code. This family was recently studied by Y.~Gao, Q.~Yue, X.~Huang, and J.~Zhang in~\cite{GYHZ} where the authors describe conditions, using the properties of classical Goppa codes, so the dual of a GRS code via a Goppa code is again a GRS code via a Goppa code. Thus, having the control over the dual, the authors are able to find the hull and give applications to quantum, LCD, self-orthogonal, and self-dual codes. Then we introduce the tensor product of GRS codes and the tensor product of GRS codes via a Goppa code. The former has been studied before due to their decoding properties. In~\cite{listtensor}, the authors provide a list decoding algorithm for the tensor product of GRS codes. In~\cite{campsmoreno2021decoding}, they authors use the tensor product of GRS codes to decode {\it hyperbolic codes}, which are augmented Reed-Muller codes, in the sense that the dimension is greater than or equal, but the minimum distance is the same. The tensor product of GRS codes via a Goppa code is important in obtaining the following result, which is proved in Section~\ref{mvGoppa_section}. Given a multivariate Goppa code $\Gamma(\cal S, g)$, $$\Gamma(\cal S, g) = \{ {\bm c}\in \F_q^n : \TGRSG{\bm c}^T = 0 \},$$
where $\TGRSG$ is a generator matrix of certain tensor product of GRS codes via a Goppa code that depends of $g.$ As a consequence, we see that the Goppa code $\Gamma(\cal S, g)$ is a subfield subcode of the dual of a tensor product of GRS codes via Goppa codes.

In Section~\ref{relations_section}, we prove that the multivariate Goppa code $\Gamma(\cal S, g)$ is a subfield subcode of an augmented Cartesian code. Augmented Cartesian (ACar) codes are a family of evaluation codes recently introduced in~\cite{lmv,aug_isit} where the authors present linear exact repair schemes for the ACar and give examples where ACar codes provide a lower bandwidth (resp., bitwidth) than RS codes (resp., Hermitian) codes when the dimension and basefield are fixed. Even more, ACar are decreasing monomial-Cartesian codes, which have applications to certain polar codes~\cite{CLMS}. In this paper, we demonstrate that 
\[\Gamma(\cal S, g)=\acar\left(\mathcal{S},{\displaystyle g}\right)_q,\]
where $\acar\left(\mathcal{S},{\displaystyle g}\right)_q$ represents the subfield subcode of certain ACar code which yields information about the basic parameters of the multivariate Goppa code $\Gamma(\cal S, g).$

In Section~\ref{hulls_section}, we study the three families just described: multivariate Goppa codes, tensor product of GRS codes via a Goppa code, and augmented Cartesian codes. Each of these families depend of a polynomial $g$ in $\F_{q^t}[{\bm x}].$ We give conditions on $g$ to determine determine subcodes, intersections, and hulls. One of the main results states that for certain $f,g$ in $\F_{q^t}[{\bm x}],$ then
\begin{itemize}
\item[\rm{(i)}] $\hull\left(\TGRSG(\cS,g)\right)=\TGRSG(\cS,\gcd(f,g))=\hull\left(\acar(\cS,g)\right),$ and
\item[\rm{(ii)}] $\Gamma(\cal S, \lcm(f,g))\subseteq \hull\left(\Gamma(\cal S, g)\right),$ with equality when $t=1.$
\end{itemize}

In Section~\ref{quantum_section}, we design quantum, LCD, self-orthogonal, and self-dual codes from multivariate Goppa codes and tensor product of GRS codes via Goppa codes, relying on the results of Section~\ref{hulls_section}. One of the main contributions in Section~\ref{quantum_section} provides an algorithm to find LCD, self-orthogonal and self-dual codes. This approach is different than that given in \cite{GYHZ}. An immediate difference is that using GRS codes, the length of the code is always bounded by the size of the field whereas this restriction is not needed in Section~\ref{quantum_section}, for instance, for the tensor product. Even more, the results of Section~\ref{quantum_section} enable a single set of defining polynomials to produce a family of codes with different lengths over a certain field (cf. \cite[Theorem 2.6]{GYHZ}). We provide some examples at the end of Section~\ref{quantum_section}. Finally, a brief summary is given as a conclusion in Section \ref{conclusion_section}.

More information about basic theory for coding theory can be found in \cite{huf-pless,MacWilliams-Sloane,van-lint}.
References for the theory of vanishing ideals and algebraic concepts used in this work are \cite{CLO1,Eisen,harris,monalg}.

\section{A parity check matrix given by the tensor product of GRS codes} \label{mvGoppa_section}
In this section, we introduce the tensor product of generalized Reed-Solomon codes via Goppa codes. We show that this family provides a parity check matrix for the multivariate Goppa codes. As a consequence, we are able to give bounds for the dimension of the Goppa code. In addition, we give a representation for the dual of a multivariate Goppa code in terms of the trace of the tensor product.

The set of $m \times n$ matrices over $\F_{q^t}$ is denoted $\F^{m \times n}_{q^t}$. The Kronecker product of matrices $A=[a_{ij}] \in \F_{q^t}^{r\times s}$ and $B \in \F_{q^t}^{m_1\times m_2}$ is the matrix that can be expressed in block form as
\[ A\otimes B := \left(\begin{array}{cccc} a_{11}B & a_{12}B & \cdots & a_{1s}B \\ a_{21}B & a_{22}B & \cdots & a_{2s}B \\
\vdots & \vdots &  & \vdots \\ a_{r1}B & a_{r2}B & \cdots & a_{rs}B \\ \end{array}\right)
\in \F_{q^t}^{rm_1\times sm_2}.
\]

\rmv{, $Row_iM$ (resp., $Col_jM$) denotes the $i^{th}$ row of $M \in \F^{m \times n}$ (resp., the $j^{th}$ column of $M$). } A generator matrix for an $[n,k,d]$ code $C$ is any matrix whose row span is $C$. Given a generator matrix $G_1$ of a code $\cC_1$ and a generator matrix $G_2$ of a code $\cC_2,$ the code $\cC_1 \otimes \cC_2$ is defined as the code whose generator matrix is $G_1\otimes G_2.$

Next, we relate the multivariate Goppa codes to generalized Reed-Solomon codes.  Given $k \in \Z^{+}$, $\F_{q^t}[x]_{<k}$ denotes the set of polynomials of degree less than $k.$ Recall that a {\it generalized Reed-Solomon} (GRS) code is defined by
\[ \GRS(S,k,g) := \left\{ \left( g(s_1)^{-1} f(s_1), \dots, g(s_n)^{-1}f(s_n) \right) : f \in \F_{q^t}[x]_{<k} \right\},\]
where $g \in \F_{q^t}[x]$ and $S \subseteq \F_{q^t}$. $\GRS$ codes in the particular case $t=1$ and $k=\deg (g)$ are called {\it GRS codes via a Goppa code} and denoted by $\GRS(S,g),$ {\it i.e.}
 \[ \GRS(S,g) := \GRS(S,\deg(g),g).\] 
 GRS codes via a Goppa code where studied in~\cite{GYHZ}. We note that $\GRS(S,k,g)$ is an $[n, k, n-k+1]$ code over $\F_{q^t}$ with $n \leq q^t$, meaning it is maximum distance separable (MDS). As we will see, tensor products of generalized Reed-Solomon codes play an important role in the duals of multivariate Goppa codes. In what follows, $n_i=|S_i|$, the cardinality of $S_i$ for $i\in[m]:=\left\{1, \ldots, m \right\}.$ From now on, when we take an element $g=g_1\cdots g_m\in\F_{q^t}[{\bm x}]$, we mean that every $g_i \in \F_{q^t}[x_i]$. The expression $g(\cS)\neq 0$ represents that $g({\bm s})\neq 0$ for all ${\bm s} \in \cS.$

\begin{definition} \label{tensor_product_codes_def}\rm
Let $\cS=S_1\times \cdots \times S_m \subseteq \F_{q^t}^m$ and $g=g_1\cdots g_m\in \F_{q^t}[{\bm x}]$ such that $g(\cS)\neq 0.$ Take ${\bm k} = \left(k_1,\ldots,k_m \right) \in \Z^m$ with $0\leq k_j \leq n_j$ for all $j\in [m].$ We define the tensor product of generalized Reed-Solomon codes as 
 \[\TGRSG(\cS,{\bm k},g):=\bigotimes_{j=1}^{m} \GRS(S_j,k_j,g_j). \] The tensor product of generalized Reed-Solomon codes via Goppa codes is 
 \[\TGRSG(\cS,g):=\bigotimes_{j=1}^{m} \GRS(S_j,\deg(g_j),g_j). \]
\rmv{
\[\TGRSG(\cS,{\bm k},g):=\bigotimes_{j=1}^{m} \GRS(S_j,k_j,g_j) \qquad \text{ and } \qquad \TGRSG(\cS,g):=\bigotimes_{j=1}^{m} \GRS(S_j,\deg(g_j),g_j). \]}
\end{definition}

A generator matrix of $\TGRSG(\cS,g)$ may be specified entrywise by
\begin{equation} \label{gen_matrix_T}
\left(  \frac{\bm{s}_i^{\bm{a}}}{g(\bm{s}_i)} \right)_{a,j} \in \F_{q^t}^{\deg (g) \times n} 
\end{equation}
where the rows and columns are indexed by $\bm{a} \in \N^{\deg(g_1) - 1 \times \cdots \times \deg( g_m) -1}$ and $i \in [n],$ respectively. 

\begin{remark}\label{21.10.30}
Observe that the tensor product of generalized Reed-Solomon codes $\TGRSG(\cS,{\bm k},g)$ has the following basic parameters.
\begin{itemize}
\item[\rm (i)] Length $n=\mid \mathcal{S} \mid.$
\item[\rm (ii)] Dimension $k= \prod_{j = 1}^m k_j.$
\item[\rm (iii)] Minimum distance $d= \prod_{j = 1}^m (n_j - k_j +1).$
\end{itemize}
In particular, $\TGRSG(\cS,g)$ is an $[n, \deg( g), \prod_{j = 1}^m (n_j - \deg(g_j) +1) \}]$ code over $\F_{q^t}.$
\end{remark}

\begin{remark}\label{21.09.15}
Note that $\GRS(S_j,k_j,g_j)=\{ \bm{0} \}$ if and only if $k_j=0.$ Thus, $\TGRSG(\cS,{\bm k},g)=\{ \bm{0} \}$ if and only if there is $j\in [m]$ such that $k_j=0.$ In addition, $\GRS(S_j,k_j,g_j)=\F_{q^t}^{n_j}$ if and only if $k_j = n_j.$ Thus, $\TGRSG(\cS,{\bm k},g)=\F_{q^t}^{n}$ if and only if ${\bm k}=(n_1,\ldots,n_m).$
\end{remark}

To relate multivariate Goppa codes to those codes in Definition \ref{tensor_product_codes_def}, observe that given any two polynomials 
$p(x_1)=p_{\ell} x_1^{\ell} + \dots+p_1 x_1 + p_0 =(x_1^{\ell},\ldots,x_1,1)\cdot (p_{\ell},\ldots,p_1,p_0)\in \F_{q^t}[x_1]$
and
$q(x_2)=q_{k} x_2^{k} + \dots+q_1 x_2 + q_0 = (x_2^{k},\ldots,x_2,1)\cdot (q_{k},\ldots,q_1,q_0)\in \F_{q^t}[x_2],$ we may abuse notation and write
\begin{eqnarray*}
p(x_1)q(x_2)=\left(\left(\begin{array}{c}
x_1^{\ell}\\
\vdots\\
x_1\\
1
\end{array}\right)
\otimes
\left(\begin{array}{c}
x_2^{k}\\
\vdots\\
x_2\\
1
\end{array}\right)\right)^T
\left(\left(\begin{array}{c}
p_\ell\\
\vdots\\
p_1\\
1
\end{array}\right)
\otimes
\left(\begin{array}{c}
q_k\\
\vdots\\
q_1\\
1
\end{array}\right)
\right).
\end{eqnarray*}
In addition, if $s\in \F_{q^t},$ then, modulo $q(x_2),$ the following two equations are valid
\begin{eqnarray}
\frac{1}{\left(x_2-s\right)}
&=& \frac{(-1)}{q(s)}\frac{\left(q(x_2)-q(s)\right)}{\left(x_2-s\right)}\\
&=& \frac{(-1)}{q(s)}
\left(\begin{array}{c}
x_2^{k-1}\\
\vdots\\
x_2\\
1
\end{array}\right)^T
\left(\begin{array}{cccccc}
q_{k} & 0 & \cdots & 0\\
q_{k-1} & q_k & \cdots & 0\\
\vdots & \vdots & \vdots& \vdots\\
q_1 & q_2 & \cdots & q_k\\
\end{array}\right)
\left(\begin{array}{c}
1\\
s\\
\vdots\\
s^{k-1}
\end{array}\right).\label{21.09.08}
\end{eqnarray}

We come to one of the main results of this section, which gives a representation of a multivariate Goppa code in terms of a tensor product of GRS codes.
\begin{theorem}\label{21.10.24}
Given a multivariate Goppa code $\Gamma(\cal S, g)$, $$\Gamma(\cal S, g) = \{ {\bm c}\in \F_q^n : \TGRSG{\bm c}^T = 0 \},$$
where $\TGRSG$ is a generator matrix of $\TGRSG(\cS,g)$; that is,  $\Gamma(\cal S, g)$ is a subfield subcode of the dual of a tensor product of GRS codes via Goppa codes.
\end{theorem}
\begin{proof}
According to (\ref{gen_matrix_T}), the following vectors generate the code $\TGRSG(\cS,g)$
\begin{equation}\label{21.09.11}
\left(\frac{\bm{s}_1^{\bm{a}}}{g(\bm{s}_1)}, \ldots, \frac{\bm{s}_n^{\bm{a}}}{g(\bm{s}_n)} \right)=
\left(\frac{s_{11}^{a_{1}} \cdots s_{1m}^{a_{m}} }{g_1(s_{11}) \cdots g_m(s_{1m})}, \ldots,
\frac{s_{n1}^{a_{1}} \cdots  s_{nm}^{a_{m}}  }{g_1(s_{n1}) \cdots g_m(s_{nm})} \right),
\end{equation}
where for $i\in [n],$ $\bm{s}_i=\left(s_{i1},\ldots,s_{im} \right) \in \F_{q^t}^m$ and $0\leq a_{j} < \deg (g_j)$ for $j\in [m].$ 

The proof consists of verifying that the elements in $\Gamma(\cal S, g)$ are orthogonal to the vectors shown in Equation~\ref{21.09.11}. We proceed by induction on $m.$ Consider the case $m=1.$ Assume $g(x)=\gamma_0 + \gamma_1 x+\dots + \gamma_k x^k.$ Equation~\ref{21.09.08} implies that if $\bm{c}=(c_1,\ldots,c_n)\in \Gamma(\cal S, g),$ then
\begin{eqnarray}
\sum_{i=1}^n \frac{c_i}{(x-s_{i})}
&=& \sum_{i=1}^n \frac{-c_i}{g(s_{i})}
\left(\begin{array}{c}
x^{k-1}\\
\vdots\\
x\\
1
\end{array}\right)^T
\left(\begin{array}{cccccc}
\gamma_{k} & 0 & \cdots & 0\\
\gamma_{k-1} & \gamma_k & \cdots & 0\\
\vdots & \vdots & \vdots& \vdots\\
\gamma_1 & \gamma_2 & \cdots & \gamma_k\\
\end{array}\right)
\left(\begin{array}{c}
1\\
s_i\\
\vdots\\
s_i^{k-1}
\end{array}\right)
\nonumber\\
&=&
\left(\begin{array}{c}
x^{k-1}\\
\vdots\\
x\\
1
\end{array}\right)^T
\left(\begin{array}{cccccc}
\gamma_{k} & 0 & \cdots & 0\\
\gamma_{k-1} & \gamma_k & \cdots & 0\\
\vdots & \vdots & \vdots& \vdots\\
\gamma_1 & \gamma_2 & \cdots & \gamma_k\\
\end{array}\right)
\sum_{i=1}^n \frac{-c_i}{g(s_{i})}
\left(\begin{array}{c}
1\\
s_i\\
\vdots\\
s_i^{k-1}
\end{array}\right)
\label{21.09.05}\\
&=& 0 \mod g(x). \label{21.09.06}
\end{eqnarray}
Observe that the polynomial in \rmv{in $\F_q[x]$ in Equation}~(\ref{21.09.05}) has degree $k-1$. As $\deg (g)=k,$ Equation~\ref{21.09.06} implies that the coefficients of the polynomial given in Equation~(\ref{21.09.05}) are zero. Hence, we see that
\begin{eqnarray*}
\left(\begin{array}{cccccc}
\gamma_{k} & 0 & \cdots & 0\\
\gamma_{k-1} & \gamma_k & \cdots & 0\\
\vdots & \vdots & \vdots& \vdots\\
\gamma_1 & \gamma_2 & \cdots & \gamma_k\\
\end{array}\right)
\left(\begin{array}{cccc}
\frac{1}{g(s_{1})}&\frac{1}{g(s_{2})} &\cdots & \frac{1}{g(s_{n})}\\
\frac{s_1}{g(s_{1})}&\frac{s_2}{g(s_{2})}&\cdots & \frac{s_n}{g(s_{n})}\\
\vdots & \vdots & \vdots & \vdots \\
\frac{s_1^{k-1}}{g(s_{1})} & \frac{s_2^{k-1}}{g(s_{2})} &\cdots& \frac{s_n^{k-1}}{g(s_{n})}
\end{array}\right)
\left(\begin{array}{c}
c_1\\
c_2\\
\vdots\\
c_n
\end{array}\right)
=
\left(\begin{array}{c}
0\\
0\\
\vdots\\
0
\end{array}\right).
\end{eqnarray*}
As the matrix in terms of $\gamma$'s is invertible, after we multiply both sides of previous equation by the inverse of this matrix, we see that the element $\bm{c}\in \Gamma(\cal S, g)$ is orthogonal to the vectors $\displaystyle \left(\frac{s_1^{a_1}}{g(s_1)}, \ldots, \frac{s_n^{a_1}}{g(s_n)} \right),$ where $0\leq a_1< k=\deg (g).$ These are the vectors that appear in Equation~(\ref{21.09.11}) when $m=1.$

Now we focus on the case $m=2.$ Assume $\deg(g_1)=k_1$ and $\deg(g_2)=k_2.$ \rmv{By definition, for $\cS=\{\bm{s}_1,\ldots,\bm{s}_n \} \subseteq \F_q^2,$ where $\bm{s}_i=(s_{i1},s_{i2}),$ and $g=g_1g_2,$ where $g_i \in \F_q[x_i],$ the multivariate Goppa code $\Gamma(\cal S, g)$ is given by
\[\Gamma(\cal S, g) = \left\{ (c_1,\ldots,c_n) \in \mathbb{F}_q^n : \sum_{i=1}^n \frac{c_i}{(x_1-s_{i1})(x_2-s_{i2})} = 0 \mod g({\bm x}) \right\}.\]}
By Equation~\ref{21.09.08}, there exist invertible matrices $A$ and $B,$ that depend of the coefficients of $g_1$ and $g_2,$ respectively, such that
\begin{eqnarray}
&&\sum_{i=1}^n \frac{c_i}{(x_1-s_{i1})(x_2-s_{i2})} \nonumber\\
&=& \sum_{i=1}^n \frac{c_i}{g_1(s_{i1})g_2(s_{i2})}
\left(\left(\begin{array}{c}
x_1^{k_1-1}\\
\vdots\\
x_1\\
1
\end{array}\right)
\otimes
\left(\begin{array}{c}
x_2^{k_2-1}\\
\vdots\\
x_2\\
1
\end{array}\right)
\right)^T
A\otimes B
\left(\begin{array}{c}
1\\
s_{i1}\\
\vdots\\
s_{i1}^{k_1-1}
\end{array}\right)
\otimes
\left(\begin{array}{c}
1\\
s_{i2}\\
\vdots\\
s_{i2}^{k_2-1}
\end{array}\right)
\nonumber\\
&=&
\left(\left(\begin{array}{c}
x_1^{k_1-1}\\
\vdots\\
x_1\\
1
\end{array}\right)
\otimes
\left(\begin{array}{c}
x_2^{k_2-1}\\
\vdots\\
x_2\\
1
\end{array}\right)
\right)^T
A\otimes B
\sum_{i=1}^n \frac{c_i}{g(\bm{s}_{i})}
\left(
\left(\begin{array}{c}
1\\
s_{i1}\\
\vdots\\
s_{i1}^{k_1-1}
\end{array}\right)
\otimes
\left(\begin{array}{c}
1\\
s_{i2}\\
\vdots\\
s_{i2}^{k_2-1}
\end{array}\right)
\right)
\nonumber\\
&=& 0 \mod g(\bm{x}).\nonumber
\end{eqnarray}
As $\deg_{x_1}(g)=k_1$ and $\deg_{x_2}(g)=k_2,$ the previous equation implies
\begin{eqnarray*}
A\otimes B
\sum_{i=1}^n \frac{c_i}{g(\bm{s}_{i})}
\left(
\left(\begin{array}{c}
1\\
s_{i1}\\
\vdots\\
s_{i1}^{k_1-1}
\end{array}\right)
\otimes
\left(\begin{array}{c}
1\\
s_{i2}\\
\vdots\\
s_{i2}^{k_2-1}
\end{array}\right)
\right)=
\left(\begin{array}{c}
0\\
0\\
\vdots\\
0
\end{array}\right).
\end{eqnarray*}
Multiplying both sides by the inverse $\left(A\otimes B\right)^{-1}=B^{-1}\otimes A^{-1},$ we finally obtain
\begin{eqnarray*}
\sum_{i=1}^n \frac{c_i}{g(\bm{s}_{i})}
\left(
\left(\begin{array}{c}
1\\
s_{i1}\\
\vdots\\
s_{i1}^{k_1-1}
\end{array}\right)
\otimes
\left(\begin{array}{c}
1\\
s_{i2}\\
\vdots\\
s_{i2}^{k_2-1}
\end{array}\right)
\right)=
\left(\begin{array}{c}
0\\
0\\
\vdots\\
0
\end{array}\right).
\end{eqnarray*}
We conclude that if $\bm{c}\in \Gamma(\cal S, g),$ then
$\bm{c}\cdot\left(\frac{s_{11}^{a_{1}}s_{12}^{a_{2}}}{g_1(s_{11})g_2(s_{12})},\ldots,
\frac{s_{n1}^{a_{1}}s_{n2}^{a_{2}}}{g_1(s_{n1})g_2(s_{n2})} \right) = 0,$ where $0\leq a_{j} < k_j=\deg(g_j),$ for $j\in [2].$ These are the vectors that appear in Equation~(\ref{21.09.11}), for $m=2.$ For the general case,  observe that following the steps of the case $m=2,$ we saw that $\sum_{i=1}^n \frac{c_i}{\prod_{j=1}^m(x_j-s_{ij})} = 0 \mod g({\bm x})$ implies that
\begin{eqnarray*}
\sum_{i=1}^n \frac{c_i}{g(\bm{s}_{i})}
\left(
\left(\begin{array}{c}
1\\
s_{i1}\\
\vdots\\
s_{i1}^{k_1-1}
\end{array}\right)
\otimes
\cdots
\otimes
\left(\begin{array}{c}
1\\
s_{in}\\
\vdots\\
s_{in}^{k_n-1}
\end{array}\right)
\right)=
\left(\begin{array}{c}
0\\
0\\
\vdots\\
0
\end{array}\right).
\end{eqnarray*}
From this fact, we conclude that if ${\bm c}\in \Gamma(\cal S, g),$ then ${\bm c}$ is orthogonal to the vectors that appear in Equation~(\ref{21.09.11}).
\end{proof}

Recall that given a code $C \subseteq \F_{q^t}^n,$ the subfield subcode over $\F_q$ is 
\[C_{q}:= \left\{ {\bm c\in C : {\bm c} \in \F_q^n} \right\}\] and the {\it field trace} with respect to the extension $\F_{q^t}^n/\F_q$  is defined as the map
$$
\begin{array}{llll}
tr: & \F_{q^t} & \rightarrow & \F_q \\
& a & \mapsto & a^{q^{t-1}} + \dots + a^{q^0}.
\end{array}
$$
The {\it trace code} of an $[n,k,d]$ code $C$ over $\F_{q^t}$ is defined by
\[\tr(C):=\left\{ \left(\tr(c_1),\ldots,\tr(c_n) \right) : (c_1,\ldots,c_n) \in C \right\}.\]
By \cite[Ch. 7. \textsection 7.]{MacWilliams-Sloane}, $\tr(C)$ is an $[n,k^*, d^*]$ over $\F_{q},$ where $k\leq k^* \leq tk$ and $d^*\geq d.$ According to Delsarte's Theorem \cite[Theorem 2]{Delsarte}, $C_q^{\perp} = tr \left( C^{\perp} \right)$. Putting this together with the fact that $
\Gamma(\cal S, g)=\left( \TGRSG(\cS,g)^{\perp} \right)_q,
$ as shown in Theorem \ref{21.10.24}, we obtain the following consequences. 

\begin{corollary} \label{mvG_dim}
The multivariate Goppa code $\Gamma(\cal S, g)$ has length $n=\mid \mathcal{S} \mid$ and dimension $k$ satisfying $n-t \deg( g) \leq k \leq n-\deg( g)$. Moreover, the dual is the trace code of a tensor product of generalized Reed-Solomon codes via Goppa codes, specifically,
$$
\Gamma(\cal S, g)^{\perp}=tr( \TGRSG(\cS,g)). \rmv{\text{ this is Delsarte's Theorem}}
$$
\rmv{I don't think that it is needed, but if you want to have an expression for the code itself, we can write:
$$
\Gamma(\cal S, g)=\left( \TGRSG(\cS,g)^{\perp} \right)_q,
$$ which agrees with Theorem~\ref{21.09.12}.}
\end{corollary}
\begin{example}\rm
Assume $\F_{3^2}^*=\left< a\right>$  is the multiplicative group of the finite field $\F_{3^2}$. Take $S_1=S_2=\left\{a^i : i\in[8] \right\}$ and $g_1=g_2 =x^2+a.$ Using the coding theory package~\cite{cod_package} for Macaulay2 \cite{Mac2}, and Magma \cite{magma}, we obtain that $\Gamma(\cal S, g)$ is an $[64,56,4]$ code over $\F_3$, which has parameters matching the best known linear code of length $64$ and dimension $56$ over this field  \cite{code_tables}. If we were to restrict ourselves to taking $m=1$, then $64 \leq q^t = 3^t$ requires $t \geq 4$ to obtain a code of length $64$. Furthermore, $4=\deg( g) +1$ implies $\deg( g) =3$. Consequently, we are only guaranteed that such a code has dimension $64-4\deg( g) = 64-12=52$.
\end{example}
In the next section, we will gain another perspective on the multivariate Goppa codes. It will allow us to round out Corollary \ref{mvG_dim} by describing the minimum distance of the multivariate Goppa codes.

\section{As subfield subcodes of augmented codes} \label{relations_section}
In this section,  we show that every multivariate Goppa code is a subfield subcode of an augmented Cartesian code \cite{lmv,aug_isit}. This useful property allows us to determine the minimum distance of the multivariate Goppa codes and establishes the necessary results for determining hulls in Section \ref{hulls_section}.

We review first the necessary facts on augmented Cartesian codes. For a lattice point $\bm{a}\in\N^m,$ $\bm{x}^{\bm{a}}=x_1^{a_1}\cdots x_m^{a_m}$ denotes the corresponding monomial in $\F_{q^t}[{\bm x}].$ The {\it graded-lexicographic order} $\prec$ on the set of monomials of $\F_{q^t}[{\bm x}]$ which is defined as $x_1^{a_1}\cdots x_m^{a_m}\prec x_1^{b_1}\cdots x_m^{b_m}$ if and only if $\sum_{i=1}^m a_i<\sum_{i=1}^m b_i$ or $\sum_{i=1}^m a_i=\sum_{i=1}^m b_i$ and the leftmost nonzero entry in $(b_1-a_1,\ldots,b_m-a_m)$ is positive. The ideal generated by $f_1, \dots, f_r \in \F_{q^t}[{\bm x}]$ is denoted $(f_1, \dots, f_r) \subseteq \F_{q^t}[{\bm x}]$. The subspace of polynomials of $\F_{q^t}[{\bm x}]$ that are $\F_{q^t}$-linear combinations of monomials $\bm{x}^{\bm{a}}\in\F_{q^t}[{\bm x}],$ where $\bm{a} \in \mathcal{A}\subseteq \N^m$, is denoted by $\cL(\cA),$ {\it i.e.}
\[\cL(\cA) :=\spn_{\F_{q^t}}\{\bm{x}^{\bm{a}} : \bm{a} \in \cA \}\subseteq \F_{q^t}[{\bm x}].\]
Together, the Cartesian product $\cS=\{\bm{s}_1,\ldots,\bm{s}_n\}\subseteq \F_{q^t}^m,$ the lattice points $\cA,$ and a polynomial $h\in \F_{q^t}[{\bm x}]$ such that $h(\bm{s})\neq 0,$ for all $\bm{s} \in \cS,$ define the {\it evaluation map}
\[
\begin{array}{lclc}
\ev({\cS},h) \colon & \cL(\cA) & \to & \F_{q^t}^{|\mathcal{\cS}|}\\
&f &\mapsto &
\left(\frac{\displaystyle f (\bm{s}_1)}{\displaystyle h (\bm{s}_1)},\ldots,\frac{\displaystyle f(\bm{s}_n)}{\displaystyle h(\bm{s}_n)}\right).
\end{array}
\]
The image of the evaluation map $\ev(\cS,h)(\cL(\cA)),$ called the {\it generalized monomial-Cartesian} code associated with $\cS, \cA,$ and $h,$ is denoted by $\cC(\cS,\cA,h)\subseteq \F_{q^t}^{|\cS|}$:
\begin{equation}\label{21.10.26}
\cC(\cS,\cA,h)= \left\{ \left(\frac{\displaystyle f (\bm{s}_1)}{\displaystyle h (\bm{s}_1)},\ldots,\frac{\displaystyle f(\bm{s}_n)}{\displaystyle h(\bm{s}_n)}\right) : f \in \cL(\cA) \right\}.
\end{equation}

We may assume that $\deg_{x_j} (h) < n_j$ for all $j\in[m]$. To see this,  consider the polynomial
\begin{equation}\label{03-24-18}
L_j(x_j) :=\prod_{\substack{s\in S_j}}\left(x_j-s\right)
\end{equation}for each $j\in[m]$.
By \cite[Lemma 2.3]{lopez-villa}, the {\it vanishing ideal} of $\mathcal{S},$ consisting of all polynomials of $\F_{q^t}[\bm{x}]$ that vanish on ${\mathcal{S}},$ is given by
$
I(\mathcal{S})=
\left(L_1(x_1),\ldots,L_m(x_m)\right).
$
 According to \cite[Proposition 4]{CLO1}, $\left\{L_1(x_1),\ldots,L_m(x_m)\right\}$ is a Gr\"obner basis of $I(\mathcal{S})$, relative to the graded-lexicographic order $\prec$. Let $r$ be the remainder of $h$ modulo $I(\cal S).$ As $r({\bm s_i})=h({\bm s_i})$ for all $i\in [n]$ and $\deg_{x_j}(r)<\deg(L_j)=n_j,$ we may redefine $h:=r$. Consequently,   $\deg_{x_j}(h) < n_j.$
 By the same reasoning, we will assume that the degree of each $f\in \cL(\cA)$ in $x_i$ is less than $n_i$; i.e., we consider $\cA \subseteq \prod_{i=1}^m\{0,\ldots, n_i-1 \}.$ In this case, the evaluation map $\ev(\cS,h)$ is injective. Thus, the length and rate of the monomial-Cartesian code $\cC(\cS,\cA,h)$ are given by $|\mathcal{S}|$ and $\frac{|\mathcal{A}|}{|\mathcal{S}|},$ respectively \cite[Proposition 2.1]{mcc}. If $m=1$ and $\cA=\{ 0, 1, \dots, k-1\}$, then $\cC(\cS,\cA,h)=GRS(S, k,h)$, the generalized Reed-Solomon code described in Section \ref{mvGoppa_section}.

A key characteristic of the monomial-Cartesian codes is that commutative algebra methods may be used to study them. The kernel of the evaluation map $\ev({\cS},h)$ is precisely $\cL(\cA) \cap I(\mathcal{S})$, where $I(\mathcal{S})$ is the vanishing ideal. Thus, algebraic properties of
$\F_{q^t}[{\bm x}]/\left( \cL(\cA) \cap I(\mathcal{S}) \right)$ are related to the basic parameters of $\cC(\cS,\cA,h).$
We now define the polynomial
\begin{equation}\label{09.07.21}
L({\bm x}):=\prod_{j=1}^m L_j^\prime(x_j),
\end{equation}
where $L_j^\prime(x_j)$ denotes the formal derivative of $L_j(x_j),$ defined in Equation~(\ref{03-24-18}). The polynomial $L({\bm x})$ plays an important role in determining the dual code $\cC(\cS,\cA,h)^\perp,$ which was studied in \cite{mcc} in terms of the vanishing ideal of $\mathcal{S}$ and in \cite{LopezDual} in terms of the indicator functions of $\mathcal{S}.$

Given $f_1,f_2 \in \F_{q^t}[\bm{x}]$ such that $f_2(\bm{s})\neq 0,$ for all $\bm{s} \in \cS,$ we write $\frac{\displaystyle f_1}{\displaystyle f_2} \in \F_{q^t}[\bm{x}]$ to mean the unique polynomial whose value at $\bm{s}$ is $\frac{\displaystyle f_1(\bm{s})}{\displaystyle f_2(\bm{s})},$ for all $\bm{s} \in \cS,$ and $\deg_{x_j}\left(\frac{\displaystyle f_1}{\displaystyle f_2}\right) < n_j.$ Observe that this polynomial is unique as it is a linear combination of indicator functions form by standard monomials. See \cite[Proposition 4.6 (a)]{LopezDual} for a more detailed explanation about this fact and these concepts.

\rmv{When the set $\cL(\cA)$ is closed under divisibility, meaning $\cL(\cA)$ satisfies the property that if $M\in \cL(\cA)$ and $M^\prime$ divides $M,$ then $M^\prime \in \cL(\cA),$ the code $\cC(\cS,\cA,h)$ is called {\it decreasing monomial-Cartesian} codes. This family of codes is important because... \textcolor{red}{write more about this}.} We now define augmented Cartesian codes, a particular family of decreasing monomial-Cartesian codes meaning that they are defined by $\cA$ such that if $M\in \cL(\cA)$ and $M^\prime$ divides $M,$ then $M^\prime \in \cL(\cA)$ \cite{CLMS}.

\begin{definition}\label{21.10.21} \rm
Let $\cS\subseteq \F_{q^t}^m$ and $\displaystyle h \in \F_{q^t}[\bm{x}]$ be  such that $h(\bm{s})\neq 0,$ for all $\bm{s} \in \cS$. An {\it augmented Cartesian code} ($\acar$ code) is defined by
\[ \acar(\mathcal{S}, {\bm k},h) := \mathcal{C}(\mathcal{S}, \cA_{Car}({\bm k}),h), \]
where $\bm{k} = (k_1, \dots, k_m),$ with $ 0 \leq k_j \leq n_j,$ and
\[\mathcal{A}_{Car}(\bm{k}) := \prod_{j = 1}^m \left\{0,\ldots,n_j-1\right\}  \setminus \prod_{j = 1}^m\left\{k_j,\ldots,n_j-1\right\}.\]
We also define
\[ \acar\left(\mathcal{S},{\displaystyle h}\right) := \acar\left(\mathcal{S}, {\bm k}_h,\frac{\displaystyle L}{\displaystyle h}\right),\]
where ${\bm k}_h := \left(n_1-\deg_{x_1}(h), \ldots, n_m-\deg_{x_m}(h) \right).$
\end{definition}
The augmented Cartesian code $\acar(\mathcal{S}, {\bm k},h),$ where $h \in \F_{q^t} \setminus \{ 0 \},$ was recently introduced and studied in \cite{lmv} due to its local properties. An augmented Cartesian code is shown in Example~\ref{21.01.01}.

\begin{remark}\label{21.10.22}
Observe that $\acar(\mathcal{S}, {\bm k},h) = \F_{q^t}^n$ if and only if $k_j=n_j,$ for some $j\in [m].$ In addition, $\acar(\mathcal{S}, {\bm k},h) = \{{\bm 0}\} $ if and only if ${\bm k} = {\bm 0}.$
\end{remark}
By previous remark, there are instances where $\acar(\mathcal{S}, {\bm k},h)$ may be one of the trivial spaces $\{ \bm{0} \}$ or $\F_{q^t}^n.$ In these cases, their basic parameters are also trivial. For the case when $\acar(\mathcal{S}, {\bm k},h)$ is nontrivial, we have the following result.
\begin{lemma}\label{21.09.03}
The augmented Cartesian code $\acar(\mathcal{S}, {\bm k},h)$ has the following basic parameters.
\begin{itemize}
\item[\rm (i)] Length $n=\mid \mathcal{S} \mid.$
\item[\rm (ii)] Dimension $k=\prod_{j = 1}^m n_j - \prod_{j = 1}^m (n_j - k_j).$
\item[\rm (iii)] Minimum distance $d=\min \left\{n_j - k_j + 1\right\}_{j\in [m]}.$
\end{itemize}
In particular, $\acar(\mathcal{S}, h)$ is an $[n, n-\deg (h), \min \{ \deg_{x_j} (h) +1 : j \in [m] \}]$ code over $\F_{q^t}.$ Moreover, the dual of the augmented Cartesian code is  $\acar(\mathcal{S}, {\bm k},h)^{\perp} =\mathcal{C}\left( \mathcal{S}, \mathcal{A}_{Car}^\perp(\bm{k}), \frac{L}{h} \right),$ where $\mathcal{A}_{Car}^\perp(\bm{k})= \prod_{j=1}^m\left\{0,\ldots,n_j-k_j-1\right\}.$
\end{lemma}
\begin{proof}
The length of the code is apparent from the definition.  Since $\acar(\mathcal{S}, {\bm k},h)$ is monomially equivalent to $\acar(\mathcal{S}, {\bm k},1),$ we can assume that $h=1$ for (ii) and (iii). Then (ii) is proven in \cite[Proposition 3.3]{lmv}. To prove (iii), note that Remark~\ref{21.10.22}, $\acar(\mathcal{S} {\bm k},h)$ gives $\{{\bm 0}\}$ if and only if ${\bm k} = {\bm 0}.$ Assume ${\bm k} \neq {\bm 0}.$ Following \cite[Definition 3.5]{CLMS}, a {\it generating set} of $\acar(\mathcal{S}, {\bm k},h)$ is defined as a set of monomials $\mathcal{B}\subseteq \mathcal{A}_{Car}(\bm{k})$ with the property that for every monomial $M\in \mathcal{A}_{Car}(\bm{k}),$ there exists a monomial $M^\prime \in \mathcal{B}$ such that $M$ divides $M^\prime.$ Observe that $\mathcal{B}=\left\{\displaystyle \frac{\displaystyle x_1^{n_1-1}\cdots  x_m^{n_m-1}}{x_j^{n_j-k_j}} : k_j > 0 \right\}$ is a generating set of $\acar(\mathcal{S}, {\bm k},h).$ Thus, the result follows from \cite[Theorem 3.9 (iii)]{CLMS}.
Finally, the dual is a consequence of the proof of the case $h=1$, given in \cite[Proposition 3.3]{lmv}, and \cite[Theorem 3.3]{CLMS}.
\end{proof}

\begin{example}\label{21.01.01} \rm
Take $K=\mathbb{F}_{17}$ and $h=1$. Let $S_1, S_2 \subseteq K$ with $n_1 = |S_1| = 6$ and $n_2 = |S_2| = 7$. The code $\acar(S_1 \times S_2, (2,2),h)$ is generated by the vectors $\ev(S_1 \times S_2,h)(\text{\textcolor{red}{$M$}}),$ where \textcolor{red}{$M$} is a point in Figure~\ref{ExACar1} (a). The dual code $\acar(S_1 \times S_2, (2,2),h)^{\perp}$ is generated by the vectors $\ev(S_1 \times S_2,L)(\text{\textcolor{blue}{$M$}}),$ where \textcolor{blue}{$M$} is a point in Figure~\ref{ExACar1} (b).
\begin{figure}[h]
\vskip 0cm
\noindent
\begin{minipage}[t]{0.45\textwidth}
\begin{center}
\begin{tikzpicture}[scale=0.55]
\draw [-latex] (0,0) -- (5.5,0)node[right] {$A_1$};
\draw [dashed] (0,1)node[left]{1} -- (5,1)node[right] {};
\draw [dashed] (0,2)node[left]{2} -- (5,2)node[right] {};
\draw [dashed] (0,3)node[left]{3} -- (5,3)node[right] {};
\draw [dashed] (0,4)node[left]{4} -- (5,4)node[right] {};
\draw [dashed] (0,5)node[left]{5} -- (5,5)node[right] {};
\draw [dashed] (0,6)node[left]{6} -- (5, 6)node[right] {};
\draw [-latex] (0,0)node[below left]{0} -- (0,6.5)node[above] {$A_2$};
\draw [dashed] (1,0)node[below]{1} -- (1,6)node[right] {};
\draw [dashed] (2,0)node[below]{2} -- (2,6)node[right] {};
\draw [dashed] (3,0)node[below]{3} -- (3,6)node[right] {};
\draw [dashed] (4,0)node[below]{4} -- (4,6)node[right] {};
\draw [dashed] (5,0)node[below]{5} -- (5,6)node[right] {};

\foreach \i in {0,...,5}
{\foreach \j in {0,...,1}
{\fill [color=red](\i,\j) {circle(.2cm)};}}
\foreach \i in {0,...,1}
{\foreach \j in {2,...,6}
{\fill [color=red](\i,\j) {circle(.2cm)};}}
\end{tikzpicture}
\vskip 0cm
(a)
\end{center}
\end{minipage}
\begin{minipage}[t]{0.45\textwidth}
\begin{center}
\begin{tikzpicture}[scale=0.55]
\draw [-latex] (6,6) -- (6,-0.5)node[below] {$A_2$};
\draw [dashed] (1,0)node[left]{} -- (6,0)node[right] {6};
\draw [dashed] (1,1)node[left]{} -- (6,1)node[right] {5};
\draw [dashed] (1,2)node[left]{} -- (6,2)node[right] {4};
\draw [dashed] (1,3)node[left]{} -- (6,3)node[right] {3};
\draw [dashed] (1,4)node[left]{} -- (6,4)node[right] {2};
\draw [dashed] (1,5)node[left]{} -- (6,5)node[right] {1};
\draw [dashed] (1,6)node[left]{} -- (6,6)node[right] {};
\draw [-latex] (6,6)node[above right]{0} -- (0.5,6)node[left] {$A_1$};
\draw [dashed] (1,0)node[below]{} -- (1,6)node[above] {5};
\draw [dashed] (2,0)node[below]{} -- (2,6)node[above] {4};
\draw [dashed] (3,0)node[below]{} -- (3,6)node[above] {3};
\draw [dashed] (4,0)node[below]{} -- (4,6)node[above] {2};
\draw [dashed] (5,0)node[below]{} -- (5,6)node[above] {1};
\draw [dashed] (6,0)node[below]{} -- (6,6)node[above] {};
\foreach \i in {3,...,6}
{\foreach \j in {2,...,6}
{\fill [color=blue](\i,\j) {circle(.2cm)};}}
\end{tikzpicture}
\vskip 0cm
(b)
\end{center}
\end{minipage}
\caption{ The code $\acar(S_1 \times S_2, (2,2),h),$ with $h=1$ and $K=\mathbb{F}_{17}$ in Example \ref{21.01.01} is generated by the vectors the vectors $\ev(S_1 \times S_2,h)(\text{\textcolor{red}{$M$}}),$ where \textcolor{red}{$M$} is a point in (a). The dual code $\acar(S_1 \times S_2, (2,2),h)^{\perp}$ is generated by the vectors $\ev(S_1 \times S_2,L)(\text{\textcolor{blue}{$M$}}),$ where \textcolor{blue}{$M$} is a point in (b).}
\label{ExACar1}
\end{figure}
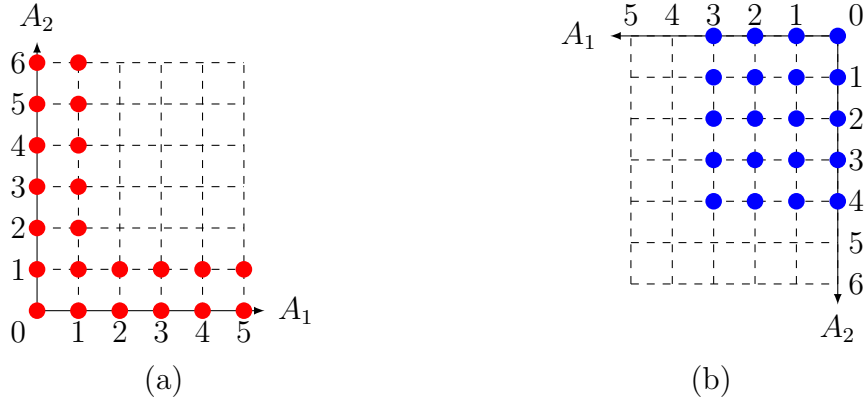
\end{example}

We now show that the dual of the tensor product of generalized Reed-Solomon codes is an augmented Cartesian code.
\begin{theorem}\label{21.09.10}
Given a tensor product of {generalized} RS codes $\TGRSG(\cS,{\bm k},g)$, its dual is
\[\TGRSG(\cS,{\bm k},g)^\perp =
\acar\left(\mathcal{S}, {\bm k}',\frac{L}{g}\right),\]
where $\bm{k}' := (n_1-k_1, \dots, n_m-k_m)$. In particular,
$\TGRSG(\cS,g)^\perp = \acar\left(\mathcal{S},{\displaystyle g}\right).$
\end{theorem}
\begin{proof}
By Lemma~\ref{21.09.03}, the dual of the augmented Cartesian code $\acar\left(\mathcal{S}, {\bm k}^\prime,\frac{\displaystyle L}{\displaystyle g}\right)$ is $\mathcal{C}(\mathcal{S}, \mathcal{A}_{Car}^\perp(\bm{k}), g ),$ where $\mathcal{A}_{Car}^\perp(\bm{k})= \prod_{j=1}^m\left\{0,\ldots,k_j-1\right\}.$ Observe that $\mathcal{C}(\mathcal{S}, \mathcal{A}_{Car}^\perp(\bm{k}), g )$ is generated by the vectors
\begin{equation*}
\left(\frac{\bm{s}_1^{\bm{a}}}{g(\bm{s}_1)}, \ldots, \frac{\bm{s}_n^{\bm{a}}}{g(\bm{s}_n)} \right)=
\left(\frac{s_{11}^{a_{1}} \cdots s_{1m}^{a_{m}}}{g_1(s_{11}) \cdots g_m(s_{1m})}, \ldots,
\frac{s_{n1}^{a_{1}} \cdots s_{nm}^{a_{m}}}{g_1(s_{n1}) \cdots g_m(s_{nm})} \right),
\end{equation*}
where for $i\in [n],$ $\bm{s}_i=\left(s_{i1},\ldots,s_{im} \right),$ and for $j\in [m],$ $0\leq a_{j} < k_j.$ The result follows from the fact that these vectors also generate $\TGRSG(\cS,{\bm k},g).$

The particular case $\TGRSG(\cS,g)^\perp = \acar\left(\mathcal{S},{\displaystyle g}\right)$ is obtained when we take $k_j=\deg(g_j),$ for $j\in [m].$
\end{proof}
\rmv{\begin{corollary}\label{21.10.23}
$\TGRSG(\cS,g)^\perp =
\acar\left(\mathcal{S},{\displaystyle g}\right).$
\end{corollary}
\begin{proof}
This is a consequence of Theorem~\ref{21.09.10} when we take $k_j=\deg(g_j),$ for $j\in [m].$
\end{proof}}

\begin{theorem}\label{21.09.12}
\rmv{Let $\cS=S_1\times \cdots \times S_m$ and $g=g_1\cdots g_m$ as in Definition~\ref{21.09.02}.}The multivariate Goppa code $\Gamma(\cal S, g)$ is the subfield subcode of an augmented Cartesian code. Specifically,
\[\Gamma(\cal S, g)=\acar\left(\mathcal{S},{\displaystyle g}\right)_q.\]
\end{theorem}
\begin{proof}
By Theorem~\ref{21.09.10}, $\acar\left(\mathcal{S},{\displaystyle g}\right)^\perp = \TGRSG(\cS,g).$
Observe that if $H$ is a parity check matrix of a code $C \subseteq \F_{q^t}^n,$ then $C_q=\left\{ {\bm c\in \F_q^n : H {\bm c}^\perp =0} \right\}.$
 Thus, the result follows from Theorem~\ref{21.10.24}.
\end{proof}

The point of view given in Theorem \ref{21.09.12} reveals additional information about the parameters of the multivariate Goppa codes, complementing Corollary \ref{mvG_dim}.

\begin{corollary}\label{21.10.27}
The multivariate Goppa code $\Gamma(\cal S, g)$ has the following basic parameters.
\begin{itemize}
\item[\rm (i)] Length $n=\mid \mathcal{S} \mid.$
\item[\rm (ii)] Dimension $k$ satisfying $n-t \deg( g) \leq k \leq n-\deg( g).$
\item[\rm (iii)] Minimum distance $d \geq \min \left\{ \deg(g_j) + 1\right\}_{j\in [m]}.$
\end{itemize}
Moreover, the dual is the trace code of a tensor product of generalized Reed-Solomon codes via Goppa codes, specifically,
$$
\Gamma(\cal S, g)^{\perp}=tr( \TGRSG(\cS,g)). \rmv{\text{ this is Delsarte's Theorem}}
$$
\end{corollary}
\begin{proof}
Given Corollary \ref{mvG_dim}, it only remains to consider the minimum distance of  $\Gamma(\cal S, g)$.
By Theorem~\ref{21.09.12} and Lemma~\ref{21.09.03} (iii), $d\geq \min \left\{\deg(g_j) + 1\right\}_{j\in [m]}.$
\end{proof}
\section{Subcodes, intersections, and hulls} \label{hulls_section}
In this section, we build on the relationships between multivariate Goppa codes, tensor products of GRS codes, and augmented Cartesian codes to determine subcodes, intersections, and hulls. To do so, we provide conditions on the defining sets of polynomials to yield the desired structures. These results will be key to the construction of entanglement-assisted quantum error correcting codes, LCD, self-orthogonal, or self-dual codes in the next section.

First, the following result helps to identify subcodes of Goppa codes, augmented Cartesian codes and tensor product of GRS codes via Goppa codes.
\begin{proposition}\label{21.09.14}
Let $g=g_1\cdots g_m,$ $g^\prime=g_1^\prime \cdots g_m^\prime \in \F_{q^t}[{\bm x}]$ be such that $g(\cS)\neq 0 \neq g^\prime(\cS).$ Then the following hold:
\begin{itemize}
\item[\rm (i)] $\TGRSG(\cS,g) \subseteq \TGRSG(\cS,gg^\prime).$
\item[\rm (ii)] $\Gamma(\cal S, gg^\prime) \subseteq \Gamma(\cal S, g).$
\item[\rm (iii)] $\acar\left(\mathcal{S},{gg^\prime}\right) \subseteq \acar\left(\mathcal{S},{g}\right).$
\end{itemize}
\end{proposition}
\begin{proof}
(i) By Equation~(\ref{21.10.26}) and the definition of a GRS code,
\begin{eqnarray*}
\TGRSG(\cS,g) &=& \left\{ \left(\frac{\displaystyle f (\bm{s}_1)}{\displaystyle g (\bm{s}_1)},\ldots,\frac{\displaystyle f(\bm{s}_n)}{\displaystyle g(\bm{s}_n)}\right) : \deg_{x_j}(f) < \deg(g_j) \right\} \\
&=& \left\{ \left(\frac{\displaystyle (fg^\prime) (\bm{s}_1)}{\displaystyle (gg^\prime) (\bm{s}_1)},\ldots,\frac{\displaystyle (fg^\prime) (\bm{s}_n)}{\displaystyle (gg^\prime) (\bm{s}_n)}\right) : \deg_{x_j}(f) < \deg(g_j) \right\} \\
&\subseteq& \left\{ \left(\frac{\displaystyle f^\prime (\bm{s}_1)}{\displaystyle (gg^\prime) (\bm{s}_1)},\ldots,\frac{\displaystyle f^\prime (\bm{s}_n)}{\displaystyle (gg^\prime) (\bm{s}_n)}\right) : \deg_{x_j}(f^\prime) < \deg(g_jg_j^\prime) \right\}\\
&=& \TGRSG(\cS,gg^\prime).
\end{eqnarray*}

(ii) By (i), $\tr\left(\TGRSG(\cS,g)\right) \subseteq \tr\left(\TGRSG(\cS,gg^\prime)\right).$ Thus, the result follows from Theorem \ref{21.10.24}.

(iii) This is a consequence of (i) and Theorem~\ref{21.09.10}.
\end{proof}

Next, we see that the intersection of multivariate Goppa codes is again a multivariate Goppa code. In addition to generalizing \cite[Theorem 3.1]{GYHZ} to multiple variables, the next result demonstrates that in order for the intersection of GRS  code via a Goppa code to be of the same type, we only require that the sum of the degrees of the defining polynomials is bounded above rather than that the two polynomials are related to one another as specified in \cite[Theorem 3.1]{GYHZ}.

\begin{theorem}\label{21.09.17}
Let $g=g_1\cdots g_m,$ $g^\prime=g_1^\prime \cdots g_m^\prime \in \F_{q^t}[{\bm x}]$ be such that $g(\cS)\neq 0 \neq g^\prime(\cS)$ and $\deg(g_j g_j^\prime) \leq n_j,$ for $j\in [m].$ Then the following hold:
\begin{itemize}
\item[\rm (i)] $\TGRSG(\cS,g) \cap \TGRSG(\cS,g^\prime) = \TGRSG(\cS,\gcd(g,g^\prime)).$
\item[\rm (ii)] $\Gamma(\cal S, g) \cap \Gamma(\cal S, g^\prime)=\Gamma(\cal S, \lcm(g,g^\prime)).$
\item[\rm (iii)] $\acar\left(\mathcal{S},{\displaystyle \lcm(g,g^\prime)}\right) \subseteq
\acar\left(\mathcal{S},{\displaystyle g}\right) \cap \acar\left(\mathcal{S},{\displaystyle g^\prime}\right)\newline \subseteq
\acar\left(\mathcal{S},{\displaystyle g}\right) + \acar\left(\mathcal{S},{\displaystyle g^\prime}\right) =
\acar\left(\mathcal{S},{\displaystyle \gcd(g,g^\prime)}\right).$
\end{itemize}
\end{theorem}
\begin{proof}
For $j\in [m],$ define $\gcd_j := \gcd(g_j,g^\prime_j)\in \F_{q^t}[x_j]$ and $\lcm_j := \lcm(g_j,g^\prime_j)\in \F_{q^t}[x_j].$ Observe that $\gcd := \gcd(g,g^\prime)= \gcd_1\cdots \gcd_m$ and $\lcm := \lcm(g,g^\prime)= \lcm_1\cdots \lcm_m.$ There are $p,p^\prime, t, t^\prime \in \F_{q^t}[{\bm x}]$ such that $g=p\gcd, g^\prime=p^\prime \gcd, \lcm(g,g^\prime)=gt$ and $\lcm(g,g^\prime)=g^\prime t^\prime.$

{\rm (i)} By Proposition~\ref{21.09.14} (i),
\[\TGRSG(\cS,\gcd) \subseteq \TGRSG(\cS,p\gcd) = \TGRSG(\cS,g)\] and
\[\TGRSG(\cS,\gcd) \subseteq \TGRSG(\cS,p^\prime\gcd)= \TGRSG(\cS,g^\prime).\]
Thus $\TGRSG(\cS,\gcd) \subseteq \TGRSG(\cS,g) \cap \TGRSG(\cS,g^\prime).$
\rmv{ Take $\bm{c}\in \TGRSG(\cS,\gcd).$ There is $f\in \F_{q^t}[{\bm x}]$ with $\deg_{x_j}(f)< \deg_{x_j}(\gcd)$ such that
\begin{eqnarray*}
\bm{c} &=& \left(\frac{f(\bm{s}_1)}{\gcd(\bm{s}_1)},\ldots, \frac{f(\bm{s}_n)}{\gcd(\bm{s}_n)} \right)=
\left(\frac{(pf)(\bm{s}_1)}{(p\gcd)(\bm{s}_1)},\ldots, \frac{(pf)(\bm{s}_n)}{(p\gcd)(\bm{s}_n)} \right)\\
&=& \left(\frac{(pf)(\bm{s}_1)}{g(\bm{s}_1)},\ldots, \frac{(pf)(\bm{s}_n)}{g(\bm{s}_n)} \right) \in
\TGRSG(\cS,g),
\end{eqnarray*}
as $\deg_{x_j}(pf) = \deg_{x_j}(p) + \deg_{x_j}(f) <  \deg_{x_j}(p) + \deg_{x_j}(\gcd) = \deg_{x_j}(g);$
\begin{eqnarray*}
\text{and }\quad \bm{c} &=& \left(\frac{f(\bm{s}_1)}{\gcd(\bm{s}_1)},\ldots, \frac{f(\bm{s}_n)}{\gcd(\bm{s}_n)} \right)=
\left(\frac{(p^\prime f)(\bm{s}_1)}{(p^\prime \gcd)(\bm{s}_1)},\ldots, \frac{(pf)(\bm{s}_n)}{(p^\prime \gcd)(\bm{s}_n)} \right)\\
&=& \left(\frac{(p^\prime f)(\bm{s}_1)}{g^\prime(\bm{s}_1)},\ldots, \frac{(p ^\prime f)(\bm{s}_n)}{g^\prime(\bm{s}_n)} \right) \in
\TGRSG(\cS,g^\prime),
\end{eqnarray*}
as $\deg_{x_j}(p^\prime f) = \deg_{x_j}(p^\prime) + \deg_{x_j}(f) <  \deg_{x_j}(p^\prime) + \deg_{x_j}(\gcd) = \deg_{x_j}(g^\prime).$ Thus we obtain that $\bm{c}\in \TGRSG(\cS,g) \cap \TGRSG(\cS,g^\prime).$}
Now take $\bm{c}\in \TGRSG(\cS,g) \cap \TGRSG(\cS,g^\prime).$ There are $f,f^\prime \in \F_{q^t}[\bm{x}]$ such that for $j\in[m], \deg_{x_j}(f) < \deg_{x_j} (g),\deg_{x_j}(f^\prime) < \deg_{x_j} (g^\prime),$ and 
\begin{equation}\label{21.10.29}
\bm{c} = \left(\frac{f(\bm{s}_1)}{g(\bm{s}_1)},\ldots, \frac{f(\bm{s}_n)}{g(\bm{s}_n)} \right)=
\left(\frac{f^\prime(\bm{s}_1)}{g^\prime(\bm{s}_1)},\ldots, \frac{f^\prime(\bm{s}_n)}{g^\prime(\bm{s}_n)} \right).
\end{equation}
Observe that $g^\prime f-gf^\prime \in I(\cS).$ As
$\deg_{x_j}(g^\prime f-gf^\prime)\leq \max \left\{\deg_{x_j}(g^\prime f),\deg_{x_j}(gf^\prime) \right\}< \deg_{x_j}(gg^\prime)\leq n_j,$ then $g^\prime f = gf^\prime.$ This implies that $\frac{g^\prime}{\gcd}f=\frac{g}{\gcd}f^\prime.$ As $\frac{g^\prime}{\gcd}$ and $\frac{g}{\gcd}$ share no common factors, $\frac{g}{\gcd}$ divides $f.$ There is $r \in \F_{q^t}[{\bm x}]$
such that $f=r\frac{g}{\gcd}.$ Thus, $g^\prime f=r\frac{gg^\prime}{\gcd}=r\lcm,$ due $\lcm\gcd=gg^\prime.$ As $\lcm$ divides $g^\prime f,$ then
\begin{eqnarray*}
\deg_{x_j}\left(\frac{g^\prime f}{\lcm}\right) &=& \deg_{x_j}(g^\prime) + \deg_{x_j}(f) - \deg_{x_j}(\lcm)\\
&=& \deg_{x_j}(\gcd) + \deg_{x_j}(\lcm) - \deg_{x_j}(g) + \deg_{x_j}(f) - \deg_{x_j}(\lcm)\\
&=& \deg_{x_j}(\gcd) - \deg_{x_j}(g) + \deg_{x_j}(f)\\
&<& \deg_{x_j}(\gcd),
\end{eqnarray*}
where the inequality holds because $\deg_{x_j}(f)<\deg_{x_j}(g).$ Equations~\ref{21.10.29} imply
\[\bm{c} = \left(\frac{(g^\prime f)(\bm{s}_1)}{(\lcm\gcd)(\bm{s}_1)},\ldots, \frac{(g^\prime f)(\bm{s}_n)}{(\lcm\gcd)(\bm{s}_n)} \right)
=\left(\frac{\left(\frac{g^\prime f}{\lcm}\right)(\bm{s}_1)}{\gcd(\bm{s}_1)},\ldots, \frac{\left(\frac{g^\prime f}{\lcm}\right)(\bm{s}_n)}{\gcd(\bm{s}_n)} \right).\]
As \rmv{$\frac{\displaystyle g^\prime f}{\displaystyle \lcm}$ is an element in $\F_{q^t}[{\bm x}]$ such that} $\deg_{x_j}\left(\frac{g^\prime f}{\lcm}\right)<\deg_{x_j}(\gcd),$ we obtain that $\bm{c}\in \TGRSG(\cS,\gcd).$ 

\rm{(ii)} By Proposition~\ref{21.09.14}~(ii),
\begin{eqnarray*}
\Gamma(\cal S, \lcm) = \Gamma(\cal S, tg) \subseteq \Gamma(\cal S, g) \end{eqnarray*}
and 
\begin{eqnarray*}\Gamma(\cal S, \lcm) = \Gamma(\cal S, t^\prime g^\prime) \subseteq \Gamma(\cal S, g^\prime).
\end{eqnarray*}
We conclude that $\Gamma(\cal S, \lcm) \subseteq \Gamma(\cal S, g) \cap \Gamma(\cal S, g^\prime).$ If $\bm{c} \in \Gamma(\cal S, g)\cap \Gamma(\cS,g^\prime),$ then
\begin{eqnarray*}
\sum_{i=1}^n \frac{c_i}{\prod_{j=1}^m(x_j-s_{ij})} = 0 \mod g({\bm x})  \end{eqnarray*}  and
\begin{eqnarray*} \sum_{i=1}^n \frac{c_i}{\prod_{j=1}^m(x_j-s_{ij})} = 0 \mod g^\prime({\bm x}).
\end{eqnarray*}
\rmv{then $\sum_{i=1}^n \frac{c_i}{\prod_{j=1}^m(x_j-s_{ij})} = 0 \mod g({\bm x})$ and $\sum_{i=1}^n \frac{c_i}{\prod_{j=1}^m(x_j-s_{ij})} = 0 \mod g^\prime({\bm x}).$} Thus, $\displaystyle \sum_{i=1}^n \frac{c_i}{\prod_{j=1}^m(x_j-s_{ij})} = 0 \mod \lcm(g,g^\prime)({\bm x}),$ which means that $\bm{c} \in \Gamma(\cal S, \lcm(g,g^\prime)).$

{\rm (iii)} By Proposition~\ref{21.09.14}~(iii),
\begin{eqnarray*}
\acar\left(\mathcal{S},{\displaystyle \lcm}\right) = \acar(\cal S, tg) &\subseteq& \acar(\cal S, g)\end{eqnarray*} { and } \begin{eqnarray*}
\acar(\cal S, \lcm) = \acar(\cal S, t^\prime g^\prime) &\subseteq& \acar(\cal S, g^\prime).
\end{eqnarray*}
This means that $\acar(\cal S, \lcm) \subseteq \acar(\cal S, g) \cap \acar(\cal S, g^\prime).$ By (i) and \cite[Ch. 1. \textsection 8.]{MacWilliams-Sloane}, $\TGRSG(\cS,g)^\perp + \TGRSG(\cS,g^\prime)^\perp = \TGRSG(\cS,\gcd)^\perp.$ Thus, by Theorem~\ref{21.09.10} we obtain 
$$\displaystyle \acar\left(\mathcal{S},{\displaystyle g}\right) \cap \acar\left(\mathcal{S},{\displaystyle g^\prime}\right) \subseteq
\acar\left(\mathcal{S},{\displaystyle g}\right) + \acar\left(\mathcal{S},{\displaystyle g^\prime}\right) =
\acar\left(\mathcal{S},{\displaystyle \gcd(g,g^\prime)}\right).$$
\end{proof}
\rmv{Let $h({\bm x})$ an element in $\F_{q^t}[{\bm x}]$ such that $h({\bm s}_i)\neq 0$ for $i\in[n].$ Define $D_h$ as the diagonal matrix $Diag\left\{h({\bm s}_1)^{-1},\ldots,h({\bm s}_n)^{-1}) \right\}.$ The following result shows that certain tensor products $\TGRSG(\cS,f)$ have the property that for some $h({\bm x}),$ $\TGRSG(\cS,f)=D_h \TGRSG(\cS,f).$} 

Let $\left( p_i {s}_i^{a} \right)_{a,i} \in \F_{q^t}^{k \times n}$ and $\left( t_i {s}_i^{a} \right)_{a,i} \in \F_{q^t}^{k \times n}$ be two generator matrices of the same GRS code, where the rows and columns are indexed by $ 0\leq a < k$ and $i \in [n],$ respectively. In \cite[Lemma 2.5]{GYHZ}, the authors describe a property that these matrices should satisfy. Specifically, if $k \leq n/2,$ then for $i\in[n],$ $p_i=\lambda t_i,$ where $\lambda \in \F_{q^t}^*.$ When $k=n,$ it is clear that the relation $p_i=\lambda t_i$ is not valid anymore, as in this case, for any coefficients $p_i, t_i,$ both matrices $\left( p_i {s}_i^{a} \right)_{a,i} \in \F_{q^t}^{k \times n}$ and $\left( t_i {s}_i^{a} \right)_{a,i} \in \F_{q^t}^{k \times n}$ generate the full  space $\F_{q^t}^n.$ The following result extends \cite[Lemma 2.5]{GYHZ} to more variables and changes the restriction from $k\leq n/2$ to $k<n.$ This result will be helpful in characterizing when the dual of a tensor product of GRS codes via Goppa codes is of the same form.

\begin{lemma}\label{21.10.28}
Let $f,$ $F \in \F_{q^t}[{\bm x}]$ be such that $f(\cS)\neq 0 \neq F(\cS).$ Define ${\bm k}:= (n_1,\ldots,n_{j^* -1}, k, n_{j^* +1},\ldots, n_m),$ where $k<n.$ Then, \rmv{$\deg(f_{j^*})<n_{j^*},$ for some  $j^*\in[m],$ and $\deg(f_{j})=n_j,$ for all $j\in [m]\setminus \{j^*\}.$ Then,}$\deg_{x_{j^*}}\left(\frac{F}{f}\right)=0$ if and only if \[\TGRSG(\cS,{\bm k},f)=\TGRSG(\cS,{\bm k},F).\]
 \rmv{If in addition $\deg(f_{j^*}) \leq \frac{n_{j^*}}{2},$ then the converse is also true.}
\end{lemma}
\begin{proof}
$(\Rightarrow)$ Assume $\deg_{x_{j^*}}\left(\frac{F}{f}\right)=0.$ \rmv{Define $h$ as the element in $\F_{q^t}[{\bm x}]$ such that $h=\frac{F}{f}.$} Recall $\ev({\cS},f)(h)=\left(\frac{\displaystyle h (\bm{s}_1)}{\displaystyle f(\bm{s}_1)},\ldots,\frac{\displaystyle h(\bm{s}_n)}{\displaystyle f(\bm{s}_n)}\right).$ Take $\ev({\cS},f)(h) \in \TGRSG(\cS,{\bm k},f).$ Then $\deg_{x_{j^*}}\left(h \right) <k.$ As $\deg_{x_{j^*}}\left(\frac{F}{f}\right)=0,$ we have that $\deg_{x_{j^*}}\left(h \frac{F}{f}\right)<k,$ which means that $\ev({\cS},F)(h \frac{F}{f}) \in \TGRSG(\cS,{\bm k},F).$ Thus, $\ev({\cS},f)(h) = \ev({\cS},F)(h \frac{F}{f}) \in \TGRSG(\cS,{\bm k},F).$ Now take $\ev({\cS},F)(h) \in \TGRSG(\cS,{\bm k},F).$ Then $\deg_{x_{j^*}}\left(h \right) <k.$ As $\deg_{x_{j^*}}\left(\frac{f}{F}\right)=\deg_{x_{j^*}}\left(\frac{F}{f}\right)=0,$ we have that $\deg_{x_{j^*}}\left(h \frac{f}{F}\right)<k,$ which means that $\ev({\cS},f)(h \frac{f}{F}) \in \TGRSG(\cS,{\bm k},f).$ Thus, $\ev({\cS},F)(h) = \ev({\cS},f)(h \frac{f}{F}) \in \TGRSG(\cS,{\bm k},f).$

$(\Leftarrow)$ Assume $\TGRSG(\cS,{\bm k},f)=\TGRSG(\cS,{\bm k},F).$ As $\ev({\cS},f)(1), \ldots, \ev({\cS},f)(x_{j^*}^{k-1}) \in \TGRSG(\cS,{\bm k},f)=\TGRSG(\cS,{\bm k},F),$ there are $\lambda_p^\ell \in \F_{q^t}[{\bm x}],$ with $p,t \in\{0,\ldots,k-1\},$ such that
\begin{eqnarray*}
\ev({\cS},f)(1)&=&\ev({\cS},F)(\lambda_0^0+\lambda^0_1x_{j^*}+\cdots+\lambda^0_{k-1}x_{j^*}^{k-1}),\\
\ev({\cS},f)(x_{j^*})&=&\ev({\cS},F)(\lambda^1_0+\lambda^1_1x_{j^*}+\cdots+\lambda^1_{k-1}x_{j^*}^{k-1}),\\
&\vdots&\nonumber\\
\ev({\cS},f)(x_{j^*}^{k-1})&=&\ev({\cS},F)(\lambda^{k-1}_0+\lambda^{k-1}_1x_{j^*}+\cdots+\lambda^{k-1}_{k-1}x_{j^*}^{k-1}),
\end{eqnarray*}
where $\deg_{x_j}(\lambda_p^\ell)< n_j$ for $j\in [m]\setminus \{j^*\}$ and $\deg_{x_{j^*}}(\lambda_p^\ell) = 0$ for all $p,t \in\{0,\ldots,k-1\}.$ Observe that for every $r\in[k-1],$
\[\ev({\cS},f)(x_{j^*}^r)
=\ev({\cS},f)(1\cdot x_{j^*}^r)=\ev({\cS},F)((\lambda^0_0+\lambda^0_1x_{j^*}+\cdots+\lambda^0_{k-1}x_{j^*}^{k-1})\cdot x^r_{j^*}).\]
Thus, for every $r\in[k-1],$
\[\ev({\cS},F)((\lambda^0_0+\lambda^0_1x_{j^*}+\cdots+\lambda^0_{k-1}x_{j^*}^{k-1})\cdot x^r_{j^*}) =
\ev({\cS},F)(\lambda^r_0+\lambda^r_1x_{j^*}+\cdots+\lambda^r_{k-1}x_{j^*}^{k-1}),\]
which means that
\[(\lambda^0_0+\lambda^0_1x_{j^*}+\cdots+\lambda^0_{k-1}x_{j^*}^{k-1})\cdot x^r_{j^*} =
\lambda^r_0+\lambda^r_1x_{j^*}+\cdots+\lambda^r_{k-1}x_{j^*}^{k-1} \mod I(\cS).\]
Define $h_r:=(\lambda^0_0+\lambda^0_1x_{j^*}+\cdots+\lambda^0_{k-1}x_{j^*}^{k-1})\cdot x^r_{j^*}$ and
$h^\prime_r:=\lambda^r_0+\lambda^r_1x_{j^*}+\cdots+\lambda^r_{k-1}x_{j^*}^{k-1}.$
Recall that the generators of the vanishing ideal $I(\cS)$ have degree $n_j$ respect to $x_j,$ for $j \in [m].$ As $\deg_{x_j}(\lambda_p^\ell)< n_j$ and  $\deg_{x_j}(h_r),$ $\deg_{x_j}(h^\prime_r)< n_j$ for $r\in[k-1]$ and $j\in[m]\setminus \{j^*\}.$ We can also see that $\deg_{x_{j^*}}(h^\prime_r) < k <  n_{j^*}$ for $r\in[k-1].$ Thus, in order to be able to compare $h_r$ and $h^\prime_r,$ we just need to know $\deg_{x_{j^*}}(h_r).$

As $\deg_{x_j^*}(h_1) = k < n_{j^*},$  $h_1=h_1^\prime.$ Thus, $\lambda_{k-1}^0=0.$
As $\lambda_{k-1}^0=0,$  $\deg_{x_j^*}(h_2) = k < n_{j^*}.$ This implies that $h_2=h_2^\prime.$ Thus, $\lambda_{k-2}^0=0.$
By induction, we see that $\lambda_{k-1}^0 = \lambda_{k-2}^0 = \cdots = \lambda_{2}^0.$ As a consequence, $\deg_{x_j^*}(h_{k-1}) = k < n_{j^*}.$ Thus, $h_{k-1}=h_{k-1}^\prime,$  which means that $\lambda_{1}^0=0.$ We conclude that $\ev({\cS},f)(1)=\ev({\cS},F)(\lambda_0^0).$ Then, $\frac{F}{f}=\lambda_0^0,$ from which we get that $\deg_{x_{j^*}}\left(\frac{F}{f}\right)=0.$
\rmv{
Because the degrees of the polynomials on both sides of previous equation, and the degrees of the generator of the vanishing ideal $I(\cS),$ 
\[(\lambda_0+\lambda_1x_{j^*}+\cdots+\lambda_{k-1}x_{j^*}^{k-1})\cdot x_{j^*} =
\beta_0+\beta_1x_{j^*}+\cdots+\beta_{k-1}x_{j^*}^{k-1}.\]
We conclude 

As $\ev({\cS},f)(1), \ev({\cS},f)(x_{j^*}^{k_{j^*}-1}) \in \TGRSG(\cS,{\bm k},f)=\TGRSG(\cS,{\bm k},F),$ there are $h,$ $h^\prime \in \F_{q^t}[{\bm x}]$ such that $\deg_{x_{j^*}}(h),$ $\deg_{x_{j^*}}(h^\prime)<k_{j^*},$ and
 \[\ev({\cS},f)(1)=\ev({\cS},F)(h) \quad \text{ and } \quad \ev({\cS},f)(x_{j^*}^{k_{j^*}-1})=\ev({\cS},F)(h^\prime).\]
This means that $\ev({\cS},F)(hx_{j^*}^{k_{j^*}-1} - h^\prime)=\ev({\cS},F)(hx_{j^*}^{k_{j^*}-1})-\ev({\cS},F)(h^\prime)=
\ev({\cS},f)(x_{j^*}^{k_{j^*}-1})-\ev({\cS},f)(x_{j^*}^{k_{j^*}-1})={\bm 0}.$ Thus $hx_{j^*}^{k_{j^*}-1} - h^\prime \in I(\cS).$ As $\deg_{x_{j^*}}(hx_{j^*}^{k_{j^*}-1} - h^\prime)\leq 2k_{j^*}-2 < n_{j^*},$ and $\deg_{x_{j}}(hx_{j^*}^{k_{j^*}-1} - h^\prime)<n_j,$ for all $j\in[m]\setminus \{j^*\},$ then $hx_{j^*}^{k_{j^*}-1} = h^\prime.$ Since $\deg_{x_{j^*}}(h^\prime)<k_{j^*},$  $\deg_{x_{j^*}}(h)=0.$ Observe that the equation $\ev({\cS},f)(1)=\ev({\cS},F)(h)$ implies that $h=\frac{F}{f}.$ Thus we obtain that $\deg_{x_{j^*}}\left(\frac{F}{f}\right)=\deg_{x_{j^*}}\left(h\right)=0.$}
\end{proof}
{Observe that the condition $\deg_{x_{j^*}}\left(\frac{F}{f}\right)=0$ means that there is an element $p({\bm x})$ in $\F_{q^t}[{\bm x}]$ such that $\deg_{x_{j^*}}\left(p \right)=0$ and $p({\bm s}_i)=\frac{F({\bm s}_i)}{f({\bm s}_i)},$ which happens if and only if $F-pf \in I(\cS).$ When $m=1,$ $p=\lambda \in \F_{q^t}.$ Since $\deg(F-\lambda f)<n,$  $F=\lambda f.$ Thus, for the case $m=1,$ {\it i.e.} only one variable, if $\TGRSG(\cS, k,f)=\TGRSG(\cS,k,F)$ and $k<n,$ then $F=\lambda f,$ which is \cite[Lemma 2.5]{GYHZ} without the restriction $k\leq \frac{n}{2}.$}

By Remark~\ref{21.09.15}, if $\TGRSG(\cS,g)$ is one of the trivial spaces $\{ \bm{0} \}$ or $\F_{q^t}^n,$ then the dual is also a tensor product of GRS codes via Goppa codes. For the case when $\TGRSG(\cS,g)$ is nontrivial, we have the following result.

\begin{theorem}\label{21.09.16}
Given $g=g_1\ldots g_m \in \F_{q^t}[{\bm x}],$ there exists $f=f_1\ldots f_m \in \F_{q^t}[{\bm x}]$ such that
\[\TGRSG(\cS,g)^\perp=\TGRSG(\cS,f),\]
if and only for some $j^*\in [m],$ the following hold:
\begin{itemize}
\rmv{\item[\rm (i)] $\deg(g_{j^*})\geq n_{j^*}/2,$}
\item[\rm (i)] $\deg(f_{j^*} g_{j^*})=n_{j^*},$
\item[\rm (ii)] $\deg(f_j)=\deg(g_j)=n_j,$ for all $j\in[m]\setminus \{j^*\},$ and
\rmv{\item[\rm (i)] $g=g_1\cdots g_m,$ $f=f_1\cdots f_m,$}
\item[\rm (iii)]$\deg_{x_{j^*}}\left(\frac{fg}{L}\right)=0.$
\rmv{\item[\rm (iii)] $fg-Lp \in I(\cS),$ for some $p\in \F_{q^t}[{\bm x}]$ such that $\deg_{x_{j^*}}\left(p\right)=0.$}
\rmv{\item[\rm (ii)] $\deg(f_{j^*} g_{j^*})=n_{j^*},$ for some $j^*\in [m],$ and
for all $p\in \F_{q^t}[{\bm x}]$ such that}
\end{itemize}

\end{theorem}
\begin{proof}
\rmv{Assume that $\TGRSG(\cS,g)$ is nontrivial.} By Theorem~\ref{21.09.10}, we just need to check that $\TGRSG(\cS,f)=\acar(\cS,g)$ if and only if  (i)-(iii) are valid. By Definition~\ref{21.10.21}, $\acar\left(\mathcal{S},{\displaystyle g}\right) = \acar\left(\mathcal{S}, {\bm k}_g,\frac{\displaystyle L}{\displaystyle g}\right),$ where ${\bm k}_g = \left(n_1-\deg(g_1), \ldots, n_m-\deg(g_m) \right).$ Thus, we will prove that $\TGRSG(\cS,f) = \acar\left(\mathcal{S}, {\bm k}_g,\frac{\displaystyle L}{\displaystyle g}\right)$ if and only if (i)-(iii) are true. Denote  the $j$-th standard vector in $\F_{q^t}^m$ by $\bm{e}_j$.

$(\Leftarrow)$ Assume (i)-(iii). By (iii), $\deg_{x_{j^*}}\left(\frac{L}{fg}\right)=\deg_{x_{j^*}}\left(\frac{fg}{L}\right)=0.$ There is $p({\bm x})\in \F_{q^t}[{\bm x}]$ such that $\deg_{x_{j^*}}\left(p \right)=0$ and $p({\bm s}_i)=\frac{L({\bm s}_i)}{(fg)({\bm s}_i)}.$ Then $\frac{L({\bm s}_i)}{g({\bm s}_i)}=(fp)({\bm s}_i),$ which means that $\deg_{x_{j^*}}\left(\frac{L}{g}\right)=\deg_{x_{j^*}}\left(f\right)=\deg\left(f_{j^*}\right).$ By (ii), ${\bm k}_g = \left(0, \ldots, n_{j^*}-\deg(g_{j^*}), \ldots, 0) \right)=\left(n_{j^*}-\deg(g_{j^*})\right)\bm{e}_{j^*}.$ Using (i),
${\bm k}_g = \deg(f_{j^*})\bm{e}_{j^*}.$ Thus, due Definition~\ref{21.10.21}, $\acar\left(\mathcal{S}, {\bm k}_g,\frac{\displaystyle L}{\displaystyle g}\right)$ is generated by the vectors
$\left(\frac{\bm{s}_1^{\bm{a}}}{\frac{L}{g}(\bm{s}_1)}, \ldots, \frac{\bm{s}_n^{\bm{a}}}{\frac{L}{g}(\bm{s}_n)} \right),$
where $0\leq a_{j} < n_j,$ for all $j\in[m]\setminus \{j^*\},$ and $0\leq a_{j^*} < \deg(f_{j^*}).$ We conclude that
for ${\bm k}:=(n_1,\ldots, n_{j^* -1}, \deg(f_{j^*}), n_{j^* +1}, \ldots,n_m)$,
$\acar\left(\mathcal{S}, {\bm k}_g,\frac{\displaystyle L}{\displaystyle g}\right)=\TGRSG\left(\cS,{\bm k},\frac{\displaystyle L}{\displaystyle g}\right).$
By (ii) $\TGRSG\left(\cS,{\bm k},f\right)=\TGRSG(\cS,f).$ Combining (iii) and Lemma~\ref{21.10.28}, we obtain
 $\acar\left(\mathcal{S}, {\bm k}_g,\frac{\displaystyle L}{\displaystyle g}\right)=\TGRSG\left(\cS,{\bm k},\frac{\displaystyle L}{\displaystyle g}\right)=\TGRSG\left(\cS,{\bm k},f\right)=\TGRSG(\cS,f).$

$(\Rightarrow)$ Assume $\TGRSG(\cS,f) = \acar\left(\mathcal{S}, {\bm k}_g,\frac{\displaystyle L}{\displaystyle g}\right),$ where ${\bm k}_g = \left(n_1-\deg(g_1), \ldots, n_m-\deg(g_m) \right).$ By Remark~\ref{21.09.15}, as $\TGRSG(\cS,f)$ is nontrivial, then $\deg(g_j)>0,$ for $j \in [m].$ According to the proof of Lemma~\ref{21.09.03}~(iii), $\mathcal{B}=\left\{\displaystyle \frac{\displaystyle x_1^{n_1-1}\cdots  x_m^{n_m-1}}{x_j^{\deg(g_j)}} : j \in [m] \right\}$ is a generating set of $\acar\left(\mathcal{S}, {\bm k}_g,\frac{\displaystyle L}{\displaystyle g}\right).$ By Definition~\ref{tensor_product_codes_def}, there is a unique {generating monomial} for $\TGRSG(\cS,f),$ meaning a monomial $\bm{x}^{\bm{a}}\in \F_{q^t}[{\bm x}]$ such that $\bm{x}^{\bm{b}}$ divides $\bm{x}^{\bm{a}}$ if and only if $\ev({\cS},f)(\bm{x}^{\bm{b}})$ is in $\TGRSG(\cS,f).$ This means that the augmented code $\acar\left(\mathcal{S}, {\bm k}_g,\frac{\displaystyle L}{\displaystyle g}\right)$ has a unique generating monomial, and it should be one of the elements in $\mathcal{B}.$ Thus, there is $j^* \in [m]$ such that $M:=\frac{\displaystyle x_1^{n_1-1}\cdots  x_m^{n_m-1}}{x_{j^*}^{\deg(g_{j^*})}}$ is the generating monomial for both $\TGRSG(\cS,f)$ and $\acar\left(\mathcal{S}, {\bm k}_g,\frac{L}{g}\right).$ As $M$ is a generating monomial of $\TGRSG(\cS,f),$ then $\deg(f_j)=n_j,$ for all $j\in[m]\setminus \{j^*\},$ and $\deg(f_{j^*}) = n_{j^*}-\deg(g_{j^*}).$ As $M$ is a generating monomial of $\acar\left(\mathcal{S}, {\bm k}_g,\frac{L}{g}\right),$ then ${\bm k}_g = \left(0, \ldots, n_{j^*}-\deg(g_{j^*}), \ldots, 0 \right),$ which implies $\deg(g_j)=n_j,$ for all $j\in[m]\setminus \{j^*\}.$ Thus, (i)-(ii) are valid and
$\TGRSG(\cS,{\bm k},f) = \TGRSG(\cS,f) = 
\acar\left(\mathcal{S}, {\bm k}_g,\frac{\displaystyle L}{\displaystyle g}\right)=\TGRSG\left(\cS,{\bm k},\frac{\displaystyle L}{\displaystyle g}\right),$ where ${\bm k}:=(n_1,\ldots, n_{j^* -1}, \deg(f_{j^*}), n_{j^* +1}, \ldots,n_m)$. By Lemma~\ref{21.10.28}, (iii) is also true.
\end{proof}
In \cite{GYHZ}, the authors use Goppa codes (the case $t=m=1$) to prove that the intersection of certain generalized Reed-Solomon codes are also generalized Reed-Solomon codes. As a consequence, they determine the hulls of certain generalized Reed-Solomon codes. Here, our focus is slightly different. Even so, taking the special case $t=m=1$ allows us to recover those results. More generally, the hull of a tensor product of generalized Reed-Solomon codes via Goppa codes is also a tensor product of generalized Reed-Solomon codes via Goppa codes, and the hull of a multivariate Goppa code contains a multivariate Goppa code (with equality when $t=1$). More precisely, we have the following result.

\begin{corollary} \label{hull_cor}
Let $\cS\subseteq \F_{q^t}^m$, $g$ and $f$ be as in Theorem~\ref{21.09.16}. Then the following hold.
\begin{itemize}
\item[\rm{(i)}] $\hull\left(\TGRSG(\cS,g)\right)=\TGRSG(\cS,\gcd(f,g))=\hull\left(\acar(\cS,g)\right).$
\item[\rm{(ii)}] $\Gamma(\cal S, \lcm(f,g))\subseteq \hull\left(\Gamma(\cal S, g)\right),$ with equality when $t=1.$
\end{itemize}
\end{corollary}
\begin{proof}
(i) By Theorems~\ref{21.09.10} and~\ref{21.09.16}, 
\begin{equation*}
\TGRSG(\cS,f)=\TGRSG(\cS,g)^\perp=\acar(\cS,g) \text{ and } \TGRSG(\cS,g)= \TGRSG(\cS,f)^\perp=\acar(\cS,f).
\end{equation*}
Thus, the result is a consequence of Theorem~\ref{21.09.17} (i).

(ii) By the poof of (i), $\TGRSG(\cS,g)=\TGRSG(\cS,f)^\perp=\acar(\cS,f).$ By Corollary \ref{mvG_dim},
$\Gamma(\cS,g)^\perp = \tr \left( \TGRSG(\cS,g) \right) = \tr \left( \acar(\cS,f) \right) \supseteq \Gamma(\cS,f).$ Thus,
\[\hull\left(\Gamma(\cal S, g)\right)= \Gamma(\cal S, g) \cap \Gamma(\cal S, g)^\perp \supseteq \Gamma(\cal S, g) \cap \Gamma(\cal S, f)
=\Gamma(\cal S, \lcm(g,f)),\]
where the last equation holds due to Theorem~\ref{21.09.17}~(ii). When $t=1,$ $\acar(\cS,f)=\Gamma(\cS,f),$ so $ \tr \left( \acar(\cS,f) \right) = \Gamma(\cS,f).$
\end{proof}
Using the conditions in Theorem \ref{21.09.16}, we can also conclude that the dimension of the Hull of the tensor product of GRS via Goppa code is
\begin{equation} \label{dim_Hull_T}
\dim \left( \hull\left(\TGRSG(\cS,g)\right) \right)= \dim \left( \TGRSG(\cS,\gcd(f,g)) \right) = \deg \left( \gcd(f,g) \right),
\end{equation}
and the dimension of the hull of the multivariate Goppa code is lower bounded by
\begin{equation} \label{dim_Hull_mvG}
\dim \left(   \hull\left(\Gamma(\cal S, g)\right) \right) {\geq} \dim \left( \Gamma(\cal S, \lcm(f,g))\right) {\geq} n - {t}\deg \left( \lcm(f,g) \right),
\end{equation}
with equality when $t=1.$

\section{Quantum, LCD, self-orthogonal and self-dual codes} \label{quantum_section}
In this section, we design entanglement quantum error-correcting codes, LCD, self-orthogonal, and self-dual codes from multivariate Goppa codes and tensor product of GRS codes via a Goppa code relying on the hulls found in the previous section.

Entanglement-assisted quantum error-correcting codes, introduced in \cite{Brun_science}, utilize entangled qubits as an enabling mechanism which allows for any linear code to be used to construct a quantum error-correcting code. These codes are a departure from constructions that employ self-dual codes. Below we use the standard notation $[[n,k,d;c]]_{q}$ code to mean a $q$-ary entanglement-assisted quantum error-correcting code (EAQECC) that encodes $k$ qubits into $n$ qubits, with minimum distance $d$, and $c$ required entangled qubits. Guenda, Jitman and Gulliver \cite{good_ea}, building on work of Wilde and Brune \cite{Wilde}, showed that the shared entanglement necessary can be captured by the dimension of the hull of the linear code used. In particular, they prove the following.

\begin{lemma} \cite[Corollary 3.2]{good_ea}  \label{eaqecc_lma}
Given an $[n,k,d]$ code $C$ over $\F_q$, there exist EAQECCs with parameters
\begin{eqnarray*}
&& [[ n, k - \dim \left(Hull(C)\right), d, n-k- \dim \left(Hull(C)\right)]]_q \quad \text{ and }\\
&& [[ n, n-k - \dim \left(Hull(C)\right), d(C^{\perp}), k- \dim \left(Hull(C)\right)]]_q.
\end{eqnarray*}
\end{lemma}

\begin{proposition}\label{21.10.31}
Let $\cS\subseteq \F_{q^t}^m$, $g$ and $f$ be as in Theorem~\ref{21.09.16}. Then the code $\TGRSG(\cS,g)$ gives rise to EAQECCs with parameters
\begin{eqnarray*}
&& [[ n, \deg \left( g \right) - \deg \left( \gcd \right), \deg(f_{j^*})+1; \deg \left( f \right) - \deg \left( \gcd \right) ]]_{q^t} \quad \text{ and }\\
&& [[ n, \deg \left( f \right) - \deg \left( \gcd \right), \deg(g_{j^*})+1; \deg \left( g \right) - \deg \left( \gcd \right)]]_{q^t},
\end{eqnarray*}
where $\gcd := \gcd(g,g^\prime).$ The code $\Gamma(\cal S, g)$ gives rise to  EAQECCs with parameters
\begin{eqnarray*}
&& [[ n, \leq t(\deg(\lcm)+\deg(g))-n, \geq \deg(f_{j^*})+1;  \leq t \deg \left( \lcm \right) - \deg \left( g \right)]]_q \quad \text{ and }\\
&& [[ n, \leq t \deg \left( \lcm \right) - \deg \left( g \right), \geq \deg(g_{j^*})+1; \leq t(\deg(\lcm)+\deg(g))-n ]]_{q},
\end{eqnarray*}
where $\lcm := \lcm(g,g^\prime),$ and equalities in the parameters of the codes when $t=1.$
\end{proposition}
\begin{proof}
The first pair of quantum codes is a consequence of Lemma~\ref{eaqecc_lma}, Remark~\ref{21.10.30}, and Equation~(\ref{dim_Hull_T}). The second pair of quantum codes follows from Lemma~\ref{eaqecc_lma}, Corollary~\ref{21.10.27}, and Inequality~(\ref{dim_Hull_mvG}).
\end{proof}
Note that when $t=1,$ which means that $\cS \subseteq \F_{q}^m,$ the two pairs of $q$-ary entanglement-assisted quantum error-correcting codes presented in Proposition~\ref{21.10.31} coincide. This happens becase in this case, from Corollary~\ref{mvG_dim}, we have that
$\Gamma(\cal S, g)^{\perp}=tr( \TGRSG(\cS,g))=\TGRSG(\cS,g),$ which means that $\TGRSG(\cS,g)=\Gamma(\cal S, f)$ and $\TGRSG(\cS,f)=\Gamma(\cal S, g).$

An $[[n,k,d;c]]_{q}$ EAQECC satisfies the Singleton Bound~\cite{Brun_science} $n + c - k  \geq 2(d - 1),$
where $0 \leq c \leq n - 1.$ The code attaching this bound is called an MDS EAQECC. As a consequence of Proposition~\ref{21.10.31}, we recover \cite[Theorem 4.5]{GYHZ}.

\begin{corollary}
Let $\cS\subseteq \F_{q}^m$, $g$ and $f$ be as in Theorem~\ref{21.09.16}. Then the code $\TGRSG(\cS,g)$ gives rise to an MDS EAQECC.
\end{corollary}
\begin{proof}
This is a consequence of Proposition~\ref{21.10.31}.
\end{proof}

Using the results of Section~\ref{hulls_section}, we now give conditions to find families of codes that are LCD, self-orthogonal, or self-dual.
\begin{corollary} \label{21.11.01}
Let $\cS\subseteq \F_{q^t}^m$, $g$ and $f$ be as in Theorem~\ref{21.09.16}. Then the following hold.
\begin{itemize}
\item[\rm{(i)}] $\TGRSG(\cS,g)$ is LCD if there exists $j \in [m]$ with $\gcd(f_j,g_j) \in \F_{q^t}$.
\item[\rm{(ii)}] $\TGRSG(\cS,g)$ is self-orthogonal if $g$ divides $f$.
\item[\rm{(iii)}] $\TGRSG(\cS,g)$ is self-dual if $f=g$.
\item[\rm{(iv)}] $\Gamma(\cal S, g)$ is LCD if $t=1$ and $\deg_{x_j}( \lcm(f,g)) {\geq} n_j$ for all $j \in [m]$.
\item[\rm{(v)}] $\Gamma(\cal S, g)$ is self-orthogonal if $t=1$ and $f$ divides $g.$
\item[\rm{(vi)}] $\Gamma(\cal S, g)$ is self-dual if $t=1$ and $f=g.$
\end{itemize}
\end{corollary}
\begin{proof}
(i) By Theorems~\ref{21.09.10} and~\ref{21.09.16}, 
\begin{equation*}
\TGRSG(\cS,f)=\TGRSG(\cS,g)^\perp=\acar(\cS,g) \text{ and } \TGRSG(\cS,g)= \TGRSG(\cS,f)^\perp=\acar(\cS,f).
\end{equation*}
Thus, the result is a consequence of Theorem~\ref{21.09.17} (i).

(ii) By the poof of (i), $\TGRSG(\cS,g)=\TGRSG(\cS,f)^\perp=\acar(\cS,f).$ By Corollary \ref{mvG_dim},
$\Gamma(\cS,g)^\perp = \tr \left( \TGRSG(\cS,g) \right) = \tr \left( \acar(\cS,f) \right) \supseteq \Gamma(\cS,f).$ Thus,
\[\hull\left(\Gamma(\cal S, g)\right)= \Gamma(\cal S, g) \cap \Gamma(\cal S, g)^\perp \supseteq \Gamma(\cal S, g) \cap \Gamma(\cal S, f)
=\Gamma(\cal S, \lcm(g,f)),\]
where the last equation holds because Theorem~\ref{21.09.17}~(ii). When $t=1,$ $\acar(\cS,f)=\Gamma(\cS,f),$ so $ \tr \left( \acar(\cS,f) \right) = \Gamma(\cS,f).$

By Remark~\ref{21.09.15}, we obtain the conditions about the LCD codes. The self-dual conditions are a consequence of the fact that when $f=g,$ then $g=\gcd(f,g)=\lcm(f,g).$
\end{proof}

\rmv{For the particular case $t=1,$ we obtain a lower and an upper bound for the dimension of the hull of the multivariate Goppa code:
\begin{equation} \label{lb_dim_Hull_mvG}
n - \deg \left( \lcm(f,g) \right) \geq \dim \left(   \hull\left(\Gamma(\cal S, g)\right) \right) \geq n - \hir{t}\deg \left( \lcm(f,g) \right).
\end{equation}}

Corollary~\ref{21.11.01} gives a simple path (with some help from the coding theory package~\cite{cod_package} for Macaulay2 \cite{Mac2} or Magma \cite{magma}) to find codes with a large length that are LCD, self-orthogonal, or self-dual codes. The key steps are the following.
\begin{enumerate}
\item Give sets $S_1, S_2\subseteq \F_{q^t}$ of cardinalities $n_1$ and $n_2,$ respectively.
\item Define $L_i:=\prod_{s\in S_i}(x-s) \in \F_{q^t}[x].$ Find the formal derivatives $L_i^\prime.$
\item Find $f_1,g_1\in\F_{q^t}[x]$ such that $f_1g_1=\lambda_1 L_1^\prime + \beta_1L_1, $ with $\lambda_1, \beta_1 \in \F_{q^t}.$
\item Find $f_2,g_2,p\in\F_{q^t}[x]$ such that $f_2g_2=\lambda_2 L_2^\prime + pL_2, $ with $\deg(p) = n_2.$
\end{enumerate}
Then the codes $\TGRSG(\cS,g_1g_{2,m})$ and $\Gamma(\cal S, g_1g_{2,m}),$ {where $g_{2,m}:=g_2(x_1)\ldots g_2(x_m)$,} have both length $n_1n_2^{m}.$ As $m$ is independent of the steps (1)-(4), after the appropriate polynomials have been found, codes with different lengths can be derived. {Observe that this is a different approach than given in \cite{GYHZ}. An immediate difference is that using GRS codes, the length of the code is always bounded by the size of the field. This restriction is not presented in the tensor product. Even more, the results of Section 5 enable a single set of defining polynomials to produce a family of codes with different lengths over a certain field (cf. [12, Theorem 2.6]). We show this in the following examples. 
\begin{example}[Family of long LCD codes]\rm
Assume $\F_{3^2}^*=\left< a\right>$.  Take $S_1 := \left\{0, 1, a, a^7 \right\}$ and $S_2 := \left\{1, a^5, a^7 \right\}.$ Define the polynomials $f_1:=x+1,$ $g_1:=2x^3 + a^5x^2+a^5x+1,$ and  $f_2:=g_2:=x^3+ax^2+2x.$ Then
\[f_1g_1=2L_1^\prime + 2L_1 \qquad \text{ and } \qquad f_2g_2=a^2L_2^\prime + pL_2,\]
where $p(x)=x^3 + a^5x^2 + a^2x + a^6.$ Then, for every $m\geq 0,$ define the polynomial in $m$ variables $f_{2,m}:=f_2(x_1)\ldots f_2(x_m).$ As $\gcd(f_1,g_1)=1,$ by Remark~\ref{21.10.30} and Corollary~\ref{hull_cor}, the tensor product $\TGRSG(\cS,f_1f_{2,m})$ is a $[4\cdot3^m, 3^m]$ LCD code over $\F_{9}.$ 
\end{example}
\begin{example}[Family of long self-orthogonal codes]\rm
Assume $\F_{3^2}^*=\left< a\right>.$ Take $S_1 := \left\{0, 1, 2, a \right\}$ and $S_2 := \left\{1, a^5, a^7 \right\}.$ Define the polynomials $f_1:=ax^3 + 2x^2 + a^7x + a,$ $g_1:=a^2x+1,$ and  $f_2:=g_2:=x^3+ax^2+2x.$ Then
\[f_1g_1=L_1^\prime +a^3L_1 \qquad \text{ and } \qquad f_2g_2=a^2L_2^\prime + pL_2,\]
where $p(x)=x^3 + a^5x^2 + a^2x + a^6.$ Then, for every $m\geq 0,$ define the polynomial in $m$ variables $g_{2,m}:=g_2(x_1)\ldots g_2(x_m).$ As $g_1$ divides $f_1,$ and $g_2$ divides $f_2,$ by Remark~\ref{21.10.30} and Corollary~\ref{hull_cor}, the tensor product $\TGRSG(\cS,g_1g_{2,m})$ is a $[4\cdot3^m, 3^m]$ self-orthogonal code over $\F_{9}.$ 
\end{example}
\begin{example}[Family of long self-dual codes]\rm
Assume $\F_{3^2}^*=\left< a\right>.$ Take $S_1 := \left\{a, a^2, a^3, a^5, a^6, a^7 \right\}$ and $S_2 := \left\{1, a^5, a^7 \right\}.$ Define the polynomials $f_1:=g_1:=x^3 + 2x + 2$ and  $f_2:=g_2:=x^3+ax^2+2x.$ Then
\[f_1g_1=L_1^\prime + L_1 \qquad \text{ and } \qquad f_2g_2=a^2L_2^\prime + pL_2,\]
where $p(x)=x^3 + a^5x^2 + a^2x + a^6.$ Then, for every $m\geq 0,$ define the polynomial in $m$ variables $g_{2,m}:=g_2(x_1)\ldots g_2(x_m).$ As $g_1=f_1,$ and $g_2=f_2,$ by Remark~\ref{21.10.30} and Corollary~\ref{hull_cor}, the tensor product $\TGRSG(\cS,g_1g_{2,m})$ is a $[6\cdot3^m, 3^{m+1}]$ self-dual code over $\F_{9}.$
\end{example}

\section{Conclusion} \label{conclusion_section}

In this paper, we defined multivariate Goppa codes which generalize the classical Goppa codes. Similar to classical Goppa codes, they can be described via a parity checks and as subfield subcodes of a family of evaluation codes. In particular, we show that considering tensor products of generalized Reed-Solomon codes via Goppa codes leads to a parity check matrix whose kernel restricted to the base field yields the multivariate Goppa codes. We also prove that multivariate Goppa codes are subfield subcodes of augmented Cartesian codes. These perspectives provide information about the code parameters as well as their hulls. As a consequence, we obtain some entanglement- assisted quantum error-correcting, LCD, self-orthogonal, and self-dual codes. We leave it as an exercise for the interested reader to translate the results in this paper to expurgated subcodes of multivariate Goppa codes.  


\end{document}